\documentclass[10pt, conference, letterpaper]{IEEEtran}
\ifCLASSINFOpdf
\usepackage[pdftex]{graphicx}
\else
\usepackage[dvips]{graphicx}
\fi

\usepackage[pass,]{geometry}

\usepackage{amsmath,amsthm,amssymb,amsfonts}
\usepackage{epsfig}
\usepackage[tight]{subfigure}
\usepackage{graphicx}
\usepackage{cite}
\usepackage{times}
\usepackage{enumerate}
\usepackage{enumitem}
\usepackage{bm}
\usepackage{algorithmic}
\usepackage{algorithm}
\usepackage{relsize}
\usepackage[utf8]{inputenc}
\usepackage{cleveref}
\usepackage{url}
\usepackage{lipsum}
\usepackage{fancyhdr}
\usepackage{tikz}
\usetikzlibrary{calc}
\crefname{section}{§}{§§}
\Crefname{section}{§}{§§}

\newtheorem{theorem}{Theorem}
\newtheorem{lemma}{Lemma}
\newtheorem{definition}{Definition}
\newtheorem{observation}{Observation}

\IEEEoverridecommandlockouts
\begin{document}

\title{Survivability in Time-varying Networks}
\author{
\IEEEauthorblockN{Qingkai Liang and Eytan Modiano}
\IEEEauthorblockA{Laboratory for Information and Decision Systems\\Massachusetts Institute of Technology, Cambridge, MA}
\thanks{This work was supported by NSF Grants CNS-1116209 and AST-1547331.}
}

\maketitle

\begin{tikzpicture}[remember picture, overlay]
\node at ($(current page.north) + (-3in,-0.5in)$) {Technical Report};
\end{tikzpicture}

\begin{abstract}
Time-varying graphs are a useful model for networks with dynamic connectivity such as vehicular networks, yet, despite their great modeling power, many important features of time-varying graphs are still poorly understood. In this paper, we study the survivability properties of time-varying networks against unpredictable interruptions. We first show that the traditional definition of survivability is not effective in time-varying networks, and propose a new survivability framework. To evaluate the survivability of time-varying networks under the new framework, we propose two metrics that are analogous to MaxFlow and MinCut in static networks. We show that some fundamental survivability-related results such as Menger's Theorem only conditionally hold in time-varying networks. Then we analyze the complexity of computing the proposed metrics and develop approximation algorithms. Finally, we conduct trace-driven simulations to demonstrate the application of our survivability framework in the robust design of a real-world bus communication network.
\end{abstract}

\section{Introduction}
Time-varying graphs have emerged as a useful model for networks with time-varying topology, especially in the context of communication networks. Examples include vehicular ad hoc networks \cite{vehicular, trace}, space communication networks \cite{space1,space2}, mobile sensor networks \cite{mobile-sensor1,mobile-sensor2}, whitespace networks\footnote{In whitespace networks, the states of secondary links change over time due to primary users' channel reclamation/release.} \cite{whitespace1,whitespace2,whitespace3} and millimeter-wave (mmWave) networks\footnote{In a mmWave network with tunable directional antennas, the network topology could vary with the dynamic adjustment of beam directions.} \cite{mmWave}. In Figure \ref{snap}, we illustrate a simple time-varying graph and its snapshots over 3 time slots.

In many applications of time-varying networks, transmission reliability is of a great concern. For example, it is critical to guarantee transmission reliability for vehicular networks that are often used to exchange traffic and emergency information; it is also crucial to provide robustness against unexpected shadowing for mmWave networks \cite{mmWave}.
Unfortunately, time-varying networks are particularly vulnerable due to their constantly changing topology that results from different kinds of interruptions. One type of interruptions are called \emph{intrinsic} interruptions which originate from the inherent nature of the network, such as node mobility in vehicular networks. For certain types of networks, such intrinsic interruptions are often \emph{predictable}. For example, it is easy to predict the temporal patterns of topology for a time-varying network formed by either public buses \cite{vehicular, trace} or satellites \cite{space1, space2} which have fixed tours and schedules; in low-duty-cycle sensor networks \cite{duty1,duty2}, the sleep/wake-up pattern is periodic and can be predicted accurately; in whitespace networks, the states of secondary links in the next few hours can be known a prior by using the whitespace database \cite{databse}; a recent study \cite{human-mobility} also shows that human mobility has 93\% potential predictability.
In contrast, the other type of interruptions are \emph{extrinsic} and \emph{unpredictable}. For example, the predictions about the evolution of network topology are prone to errors and could be inaccurate due to various unforeseen factors such as unexpected obstacles and hardware malfunctions. These unpredictable disruptions may greatly degrade network performance and are referred to as \textbf{failures}.  The goal of this paper is to understand the robustness of time-varying networks against unpredictable interruptions (failures) while treating those predictable interruptions as an inherent feature of the network.

Due to the unpredictability of failures, it is desirable to evaluate the \emph{worst-case survivability}. In static networks, this is usually defined to be the ability to survive a certain number of failures as measured by the mincut of the graph. However, this definition is not effective in time-varying networks. By its very nature, a time-varying network may have different topologies at different instants, so its connectivity or survivability must be measured over a long time interval. To be more specific, we would like to highlight two important temporal features that are neglected by the traditional notion of survivability.

\begin{figure}[tbp]
\begin{center}
\includegraphics[width=3.3in]{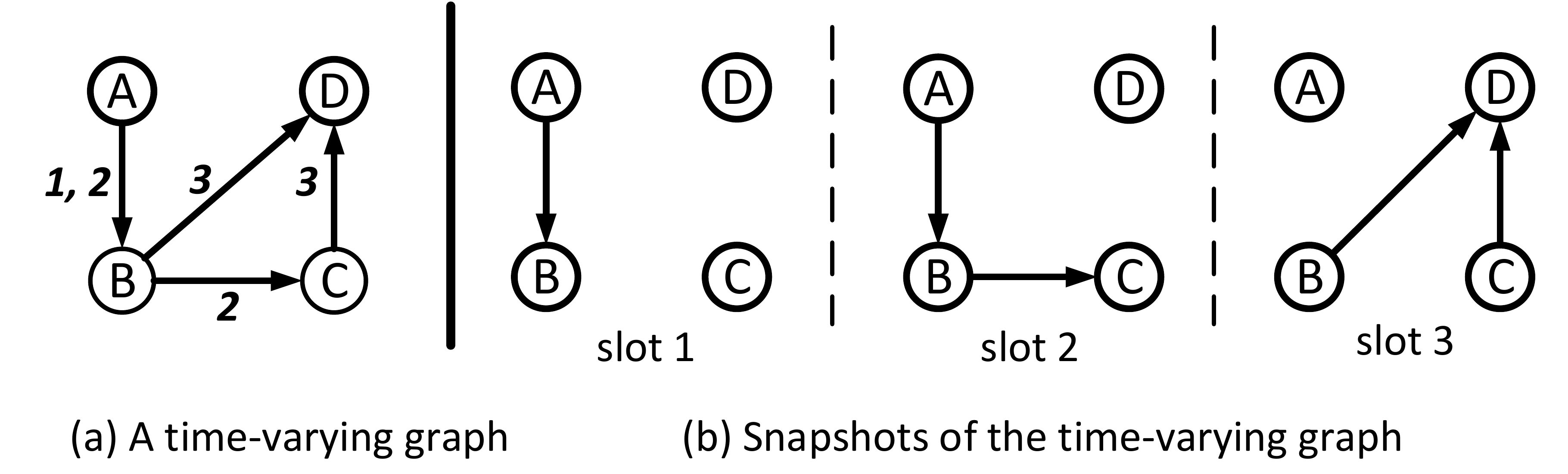}
\caption{(a) The original time-varying graph, where the numbers next to each edge indicate the slots when that edge is active. The traversal delay over each edge is one slot. (b) Snapshots of the time-varying graph.}
\label{snap}
\end{center}\vspace{-3mm}
\end{figure}

First, failures have significantly different durations in a time-varying network. For example, an unexpected obstacle may only disable the link between two nodes for several seconds, after which the link reappears. In contrast, the traditional definition of survivability is intended for a static environment and fails to account for links reappearing. The duration of failures has a crucial impact on the performance of time-varying networks; for example, in the time-varying network shown in Figure \ref{snap}, an one-slot failure of any link cannot separate node $\mathsf{A}$ and node $\mathsf{D}$ while a two-slot failure (i.e., a failure that spans two consecutive slots) can disconnect $\mathsf{D}$ from $\mathsf{A}$ by disabling link $\mathsf{A\rightarrow B}$ in the first two slots.

Second, failures may occur at different instants. This feature is obscured in static networks but has a great influence on time-varying networks due to their changing connectivity. For example, if a two-slot failure occurs to link $\mathsf{A\rightarrow B}$ at the beginning of slot 2, node $\mathsf{D}$ is still reachable from node $\mathsf{A}$ within the three slots; however, if the two-slot failure happens  at the beginning of slot 1, there is no way to travel from $\mathsf{A}$ to $\mathsf{D}$ within the three slots.

To handle the above non-trivial temporal factors, we propose a new survivability framework for time-varying networks. Our framework captures both the \emph{number} and the \emph{duration} of failures. The contributions of this paper are in the following four areas:

\noindent $\bullet$ {\textbf{Model.} We propose a new survivability framework, i.e., $(n,\delta)$-survivability, where the values of $n$ and $\delta$ characterize the number and the duration of failures the network can tolerate. Moreover, by tuning the two parameters, our framework can generalize many existing survivability models. We further propose two metrics, namely  $\mathsf{MinCut_{\delta}}$ and $\mathsf{MaxFlow_{\delta}}$, in order to assess robustness in time-varying networks. }

\noindent $\bullet$ {\textbf{Theory.} We provide new graph-theoretic results that highlight the difference between static and time-varying graphs. For example, we show that some fundamental survivability-related results such as Menger's Theorem\footnote{In graph theory, Menger's Theorem is a special case of the maxflow-mincut theorem, which states that the maximum number of edge- or node-disjoint paths equals to the size of the minimum edge or node cut, respectively.} only conditionally hold in time-varying graphs. }

\noindent $\bullet$ {\textbf{Computation.} Due to the difference between static and time-varying graphs, the evaluation of survivability becomes very challenging in time-varying networks. We analyze the complexity of computing the proposed survivability metrics and develop efficient approximation algorithms.}

\noindent $\bullet$ {\textbf{Application.} We conduct trace-driven simulations to demonstrate the application of our framework in a real-world communication network used in a public transportation system. It is shown that our survivability framework has strong modeling power and is more suitable for time-varying networks than existing approaches.}

The remainder of this paper is organized as follows. In Section \ref{graph_model}, we formalize the model of time-varying graphs. In Section \ref{survive_model}, the new survivability framework and its associated metrics are introduced. In Section \ref{comp}, we investigate some computational issues in the proposed framework. In Section \ref{app1}, trace-driven simulations are conducted to demonstrate the application of our framework in a real-world bus communication network. Finally, related work and conclusions are given in Sections \ref{related_works} and \ref{conclusion}, respectively.
\section{Model of Time-varying Graphs}\label{graph_model}
In this section, we formalize the model of time-varying graphs and introduce some important terminology and assumptions that will be frequently used throughout the paper. A useful tool for transforming time-varying graphs is also introduced.
\subsection{Definitions and Assumptions}\label{def}
Time-varying graphs are a high-level abstraction for networks with time-varying connectivity. The formal definition, first proposed in \cite{TVG2}, is as follows.

\begin{definition}[{Time-Varying Graph}]\label{def_time}
A time varying graph $\mathcal{G}=(G,\mathcal{T},\rho,\zeta)$ has the following components: \\
(i) Underlying (static) digraph $G=(V,E)$;\\
(ii) Time span $\mathcal{T}\subseteq \mathbb{T}$, where $\mathbb{T}$ is the time domain;\\
(iii) Edge-presence function $\rho: E\times \mathcal{T}\mapsto \{0,1\}$, indicating whether a given edge is active at a given instant; \\
(iv)  Edge-delay function $\zeta: E\times \mathcal{T}\mapsto \mathbb{T}$, indicating the time spent on crossing a given edge at
a given instant.
\end{definition}

\noindent This model can be naturally extended by adding a node-presence function and a node-delay function. However, it is trivial to transform node-related functions to edge-related functions by the technique called node splitting (see \cite{LP}, Chapter 7.2); thus, it suffices to consider the above edge-version characterization.

In this paper, we consider a \emph{discrete and finite} time span, i.e., $\mathcal{T}=\{1,2,\cdots,T\}$, where $T$ is a bounded integer indicating the time horizon of interests, measured in the number of slots. In practice, $T$ may have different physical meanings. For instance, it may refer to the deadline of packets or delay tolerance in delay-tolerant networks; it may also correspond to the period of a network whose topology varies periodically (e.g., satellite networks with periodical orbits). The slot length of a time-varying graph is arbitrary as long as it can capture topology changes in sufficient granularity.
%For example, if link states change frequently due to fast fading, the slot length needs to be very small (milliseconds) to capture such rapid variations; in contrast, if the network topology changes due to node mobility, the slot length may be much larger (seconds or minutes). In general, a smaller slot length yields more modeling accuracy but increases computational complexity.
% \footnote{In such a discrete-time system, ``time $t$" just means ``slot $t$", and we will use ``time" and ``slot" interchangeably.}

Under the discrete-time model, the edge-delay function $\zeta$ can take values from $\mathbb{N}=\{0,1,\cdots\}$. Note that \emph{zero} delay means that the time used for crossing an edge is negligible as compared to the slot length.
%and multiple hops can be traversed within a single slot. %Also note that the traversal time can be different for different edges in different slots.%
Throughout the rest of this paper, we consider the case where edge delay is one slot, i.e., $\zeta(e,t)=1$ for any $e\in E$ and $t\in \mathcal{T}$, however, it is trivial to extend the analysis to arbitrary traversal time.

The edge-presence function $\rho$ indicates the \emph{predictable} topology changes in a time-varying network. Examples of such predictable topology changes include those in a space communication network  with known orbits\cite{space1,space2}, in a mobile social network  consisting of students who
share fixed class schedules\cite{mobile-social}, in a low-duty-cycle sensor network with periodic sleep/wake-up patterns \cite{duty1,duty2}, in a whitespace network with planned channel reclamation\cite{whitespace1,whitespace2}, in a mmWave network  with scheduled beam steering\cite{mmWave}, etc. In contrast, \emph{unpredictable} topology changes (also referred to as \textbf{failures} in this paper) include those caused by unexpected shadowing, unscheduled channel reclamation, hardware malfunctions, etc. Note that this model does not require perfect predictions of future topology changes since any prediction errors can be treated as failures.
%In Section \ref{app1}, we will demonstrate how to flexibly model predictable and unpredictable interruptions in real-world applications.
%In Sections \ref{app1} and \ref{app2}, we will further discuss two detailed real-world applications where intrinsic.
%we assume that the entire time-varying graph $\mathcal{G}=(G,\mathcal{T},\rho,\zeta)$ is given in advance. This essentially requires complete knowledge about how the network topology changes over time. In other words, we assume \emph{predictable} topology changes.

%Finally, we would like to highlight the modeling power of time-varying graphs. In practice, they can be used to model almost all types of networks with topology changes. Examples include (i) mobile networks (e.g., vehicular or satellite networks), where the network topology varies due to node mobility, (ii) whitespace networks where link states change due to primary users' channel reclamation/release, and (iii) mmWave networks with tunable directional antennas where the network topology varies with the dynamic adjustment of beam directions.
\subsection{Terminology}

\begin{definition}[Contact]
There exists a contact from node $u$ to node $v$ in time slot $t$ if $e=(u,v)\in E$ and $\rho(e,t)=1$. This contact is denoted by $(e,t)$ or $(uv, t)$.
\end{definition}

\noindent Intuitively, a contact is a ``temporal edge", indicating the activation of a certain edge in a certain time slot. In the example shown in Figure \ref{snap}, there exists a contact $\mathsf{(AB,1)}$, showing that link $\mathsf{A\rightarrow B}$ is active in slot 1. %In the rest of this paper, we will denote $C$ the set of contacts in the time-varying graph. It is useful to notice the simple bound $|C|\le |E|T$ since each edge is active for at most $T$ slots.

\begin{definition}[{Journey} \cite{TVG1}]
In a time-varying graph, a journey from node $s$ to node $d$ is a sequence of contacts:
$(e_1,t_1)\rightarrow(e_2,t_2)\rightarrow\cdots\rightarrow (e_n,t_n)$ such that for any $i<n$\\
(i) ${\mathsf{start}}(e_1)=s$, $\mathsf{end}(e_n)=d$;\\
(ii)  $\mathsf{end}(e_{i})=\mathsf{start}(e_{i+1})$; \\
(iii) $\rho(e_i,t_i)=1$;\\ % $\rho(e_i,t)=1$ for any $t$ such that $t_i\le t<t_i+\zeta(e_i,t_i)$;
(iv) $t_{i+1}> t_i$ and $t_n\le T$. % $t_{i+1}\ge t_i+\zeta(e_i,t_i)$.
\end{definition}

\noindent Intuitively, a journey is just a ``time-respecting" path. Conditions (i)-(ii) mean that intermediate edges used by a journey are spatially connected. Condition (iii) requires that intermediate edges remain active when traversed. Condition (iv) indicates that the usage of intermediate edges must respect time and the journey should be completed before the time horizon $T$.  For example, there exists a journey from $\mathsf{A}$ to $\mathsf{D}$ in Figure \ref{snap}: $\mathsf{(AB,1)\rightarrow (BC,2)\rightarrow (CD, 3)}$ when $T=3$. %Note that in the above definition, $t_1,\cdots,t_n$ are the times when the journey \emph{starts} to cross edges $e_1,\cdots,e_n$. Hence, the full set of contacts used a journey should also account for the edge traversal time and include the additional contacts spent on crossing edges. Mathematically, denote $C(J)$ the set of contacts used by journey $J$. Then $C(J)=c_1\cup c_2\cup\cdots \cup c_n$, where $c_i=\{(e_i,t)|0\le t-t_i<\zeta(e_i,t_i)\}$ is the set of contacts spent on crossing edge $e_i$.

\begin{definition}[{Reachability}]
Node $d$ is reachable from node $s$ if there is a journey from $s$ to $d$.
\end{definition}

\noindent Intuitively, reachability can be regarded as ``temporal connectivity" which indicates whether two nodes can communicate within $T$ slots. For example, node $\mathsf{D}$ is reachable from node $\mathsf{A}$ in Figure \ref{snap}, meaning that a message from $\mathsf{A}$ can reach $\mathsf{D}$ within $T=3$ slots.

\subsection{A Useful Tool: Line Graph}\label{line}
A line graph is a useful tool which allows us to transform a time-varying graph into a static graph that preserves the original reachability information. Readers may temporarily skip the details and revisit this section when necessary.

The transformation uses a similar idea to the classical \emph{Line Graph} \cite{line_ref} which illustrates the adjacency between edges. The difference here is that we also need to consider the temporal features of time-varying graphs.
Given a time-varying graph $\mathcal{G}$ with source $s$ and destination $d$, its Line Graph $L(\mathcal{G})$ is constructed as follows.

\begin{itemize}[leftmargin=0.3cm,itemsep=1mm,topsep=1mm]
\item {For each contact $(e,t)$ in the original time-varying graph $\mathcal{G}$, create a corresponding node in the Line Graph; the new node is denoted by $v_{e,t}$. In addition, create a node for the source $s$ and a node for the destination $d$, respectively.}

\item {Add a directed edge from node $v_{e_1,t_1}$ to node $v_{e_2,t_2}$ in the Line Graph if $(e_1,t_1)\rightarrow(e_2,t_2)$ is a feasible journey from $\mathsf{start}(e_1)$ to $\mathsf{end}(e_2)$. Also, add an edge from node $s$ to node $v_{e,t}$ if $\mathsf{start}(e)=s$, and add an edge from node $v_{e,t}$ to node $d$ if $\mathsf{end}(e)=d$.}
\end{itemize}

%\noindent Note that the construction of the Line Graph only relies on the original time-varying graph and the source-destination pair.
\noindent An example of the Line Graph is shown in Figure \ref{line_example}. The Line Graph is useful in the sense that it preserves the information of every $s$-$d$ journey in the original time-varying graph. In Figure \ref{line_example}, we can observe the correspondence between journey $\mathsf{(AB,1)}$ $\mathsf{\rightarrow (BC,2)}$ $\mathsf{\rightarrow (CD,3)}$  and path $\mathsf{A\rightarrow V_{AB,1}}$ $\mathsf{\rightarrow V_{BC,2} \rightarrow V_{CD,3}\rightarrow D}$. This is generalized in Observation \ref{line_corre} whose correctness is easy to verify.
\begin{observation}\label{line_corre}
Every $s$-$d$ journey in a time-varying graph has an one-to-one correspondence to some $s$-$d$ path in its Line Graph.
\end{observation}\vspace{-5mm}

%\noindent Observation \ref{line_corre} implies that we can handle many journey-related problems in a time-varying graph by looking at paths in the corresponding Line Graph. This allows us to apply path-related results in static graphs, such as Menger's Theorem and Dijkstra Algorithm, to help solve journey-related problems in time-varying graphs.

%However, it should be noted that Line Graphs are just a gadget that can facilitate our analysis but not a replacement for time-varying graphs. In many cases, the conversion to Line Graphs doesn't change the non-triviality of  time-varying graphs.

\begin{figure}[ht]
\begin{center}
\includegraphics[width=3.2in]{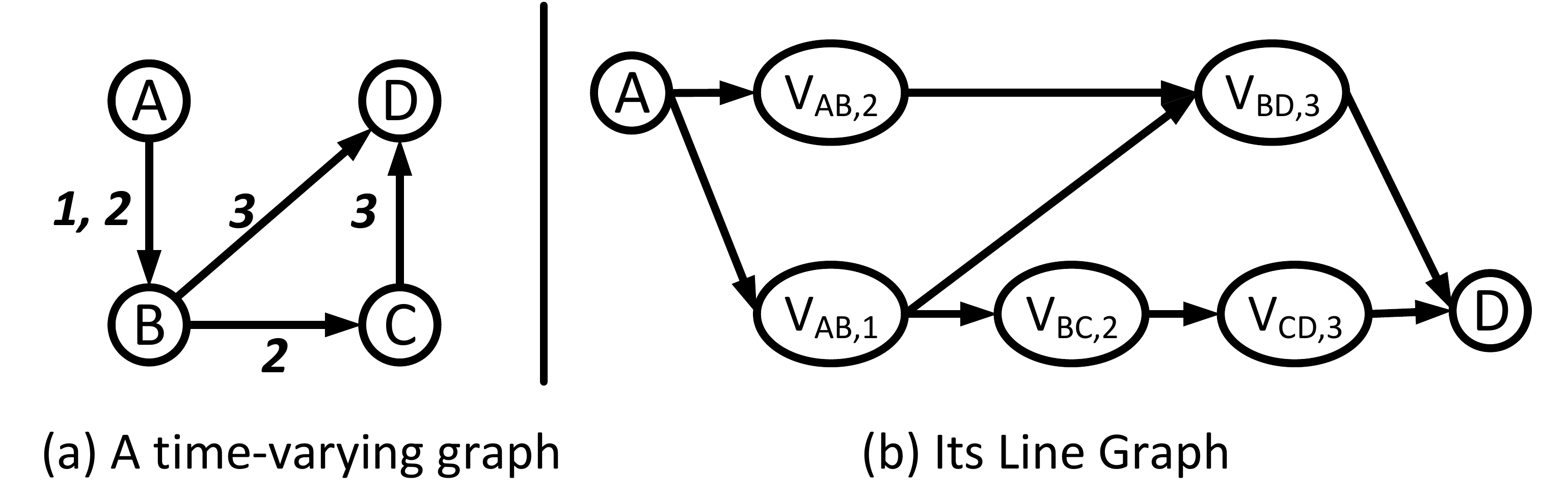}
\caption{Illustration of Line Graph (src: $\mathsf{A}$, dst: $\mathsf{D}$). }
\label{line_example}
\end{center}\vspace{-3mm}
\end{figure}

\section{Survivability Model and Metric}\label{survive_model}
In this section, we begin to investigate the survivability properties of time-varying networks. Specifically, we are interested in their resilience against \emph{unpredictable interruptions} (i.e., failures) such as unexpected shadowing, hardware malfunctions, etc. % By comparison, those \emph{predictable} topology changes are treated as the nature of time-varying networks.
% Similar to those \emph{intrinsic events} such as node mobility, exogenous interruptions can also cause topology variations. However, such extrinsic topology changes are totally \emph{unpredictable} while those intrinsic events are usually \emph{predictable} (see Sections \ref{app1}-\ref{app2} for examples).

We first develop a new survivability model for time-varying networks. Next, several metrics are introduced to evaluate  survivability under the new model. Finally, we  present some graph-theoretic results regarding these metrics, which highlights the key difference between time-varying and static networks. In particular, we will show that some fundamental survivability-related
results in static networks, such as Menger's Theorem, only conditionally hold in time-varying networks. Such a difference makes it challenging to evaluate survivability in a time-varying network.
\subsection{$\mathlarger{\mathlarger{(n,\delta)}}$-Survivability}
In static networks, the \emph{worst-case} survivability is usually defined to be the ability to survive a certain number of failures wherever these failures occur. This definition is still feasible but very ineffective in time-varying networks because it fails to capture many temporal features of failures (e.g., duration and instant of occurrence). As discussed in the introduction, these temporal features  have significant impacts on time-varying networks. Hence, we extend the survivability model in order to account for these temporal effects and propose the concept of ${(n,\delta)}$-Survivability.
We first define $(n,\delta)$-survivability \emph{for a given source-destination pair}, i.e., pairwise $(n,\delta)$-survivability.

\begin{definition}[{Pairwise ${(n,\delta)}$-Survivability}]\label{local_def}
In a time-varying graph $\mathcal{G}$, a source-destination pair $(s,d)$ is $(n,\delta)$-survivable if $d$ is still reachable from $s$ after the occurrence of any $n$ failures, with each failure lasting for at most $\delta$ slots.
\end{definition}

\noindent We can further define global $(n,\delta)$-survivability.

\begin{definition}[{Global ${(n,\delta)}$-Survivability}] A time-varying network is $(n,\delta)$-survivable if all pairs of nodes are $(n,\delta)$-survivable.
\end{definition}

\noindent Since it only takes $O(|V|^2)$ to check all pairs of nodes, global $(n,\delta)$-survivability can be easily derived from pairwise $(n,\delta)$-survivability. Therefore, we will focus on pairwise $(n,\delta)$-survivability for a given pair of nodes $(s,d)$ throughout the rest of this paper.

\vspace{1mm}

\noindent \textbf{Discussion:} The above definitions do not impose any assumption about \emph{when} and \emph{where} the $n$ failures occur and thus imply the \emph{worst-case} survivability. In other words, $(n,\delta)$-survivability means the network can survive $n$ failures that last for $\delta$ slots \emph{wherever} and \emph{whenever} these failures occur. The parameter $n$ reflects ``spatial survivability", indicating \emph{how many} failures the network can survive, and the parameter $\delta$ reflects ``temporal survivability", indicating \emph{how long} these failures can last.

Note also that $(n,\delta)$-survivability is a generalized definition. For example, if $\delta$ = $T$ (note that $T$ is the time horizon), then $(n,\delta)$-survivability reflects the number of \emph{permanent failures} the network can tolerate, which becomes the conventional notion of survivability used in static networks.
%Second, $(n,\delta)$-survivability is a generalized definition. For example, if $\delta$ = $T$ (note that $T$ is the time horizon), then $(n,\delta)$-survivability essentially considers \emph{permanent failures} which are the type of failures assumed in static networks. Since static networks are just a special case of time-varying networks, $(n,\delta)$-survivability generalizes the traditional notion of survivability in static networks.

Finally, it should be mentioned that failures can be either link failures or node failures. Since node failures can be converted to link failures by node splitting (see \cite{LP}, Chapter 7.2), we will consider link failures unless otherwise stated.
\subsection{Survivability Metrics}\label{intro_metrics}
In static networks, two commonly-used survivability metrics are: MinCut, i.e., the minimum number of edges whose deletion can separate the source and the destination, and MaxFlow, i.e., the maximum number of edge-disjoint paths from the source to the destination.  If MinCut (or MaxFlow) equals to $n$, the destination is still connected to the source after any $n-1$ link failures. However, by its very nature, a time-varying network has different topologies at different instants, so its connectivity or survivability must be measured over a long time interval and these static metrics cannot be directly applied to time-varying networks. In this section, we introduce two new metrics for $(n,\delta)$-survivability. The fundamental relationship between the two metrics will be further discussed in Section \ref{ana_metrics}.

\vspace{1mm}

\noindent \emph{\textbf{1) Survivability Metric: $\mathlarger{\mathsf{MinCut_{\delta}}}$}}

\vspace{1mm}

Before we proceed to the first survivability metric, it is necessary to introduce the notions of $\delta$-removal and $\delta$-cut.
\begin{definition}[{$\delta$-removal}]
A $\delta$-removal is the deletion of a link for $\delta$ consecutive time slots.
\end{definition}

\noindent Intuitively, a $\delta$-removal just corresponds to a link failure that lasts for $\delta$ consecutive time slots. %For the simplicity of notations, if a $\delta$-removal disables link $e$ in slots $t, t+1,\cdots,t+\delta-1$, we say that contact $(e,t)$ is the \emph{Removal Head} of this $\delta$-removal, meaning that this $\delta$-removal starts to block link $e$ in slot $t$. Any $\delta$-removal can be uniquely identified by its removal head.

\begin{definition}[{$\delta$-cut}]
A $\delta$-cut  is a set of $\delta$-removals that can render the destination unreachable from the source.
\end{definition}

\noindent  The above definition is similar to the traditional notion of graph cuts except that $\delta$-cuts also account for the duration of removals.

Now we are ready to introduce the first metric for $(n,\delta)$-survivability, namely $\mathsf{MinCut_{\delta}}$. This metric directly follows from the definition of $(n,\delta)$-survivability and is analogous to MinCut in static networks.

\begin{definition}[\textbf{${\mathsf{MinCut_{\delta}}}$}]\label{mincut_def}
$\mathsf{MinCut_{\delta}}$ is the cardinality of the smallest $\delta$-cut, i.e., the minimum number of $\delta$-removals needed to render the destination unreachable from the source.
\end{definition}

\noindent \textbf{Discussion.} First, $\mathsf{MinCut_{\delta}}$ gives the minimum number of $\delta$-removals required to disconnect the time-varying network.
%If the number of $\delta$-removals is strictly smaller than $\mathsf{MinCut_{\delta}}$, the destination is still reachable from the source.
In particular, when $\mathsf{MinCut_{\delta}}=n$, the source-destination pair can safely survive any $n-1$ failures that last for $\delta$ slots and is thus $(n-1,\delta)$-survivable.
Second, $\mathsf{MinCut_{\delta}}$ generalizes MinCut in static networks since we can simply set $\delta$ = $T$ such that a $\delta$-removal becomes a permanent link removal.

%Note that computing $\mathsf{MinCut_{\delta}}$ corresponds to an Integer Linear Programming (ILP) problem; the detailed ILP formulation is omitted for brevity and can be found in the technical report \cite{tech}.

\vspace{1mm}

\noindent \textbf{Formulation.} $\mathsf{MinCut_{\delta}}$ corresponds to the following Integer Linear Programming (ILP) problem:
\[
\begin{split}
\min~~~~~~~~&\sum_{(e,t)\in C} y_{e,t}\\
\text{s.t.}~~~~~~~~&\sum_{(e,t)\in R(\delta,J)}y_{e,t}\ge 1,~\forall J\in\mathcal{J}_{sd},\\
               &y_{e,t}\in \{0,1\},~\forall (e,t)\in C.
\end{split}
\]
Here, $y_{e,t}$ is a binary variable indicating whether a $\delta$-removal occurs to edge $e$ in slot $t$, and $C$ is the set of contacts in the time-varying graph. $\mathcal{J}_{sd}$ is the set of feasible journeys from $s$ to $d$. For any $J\in \mathcal{J}_{sd}$, we define $R(\delta,J)$ as the set of contacts $\{(e,t)\}$ such that if $y_{e,t}=1$ then journey $J$ will be disrupted, i.e., $R(\delta,J)=\{(e,t)| ~\exists (e,t')\in C_J~\text{s.t.}~0\le t'-t<\delta\}$, where $C_J$ is the set of contacts used by journey $J$. Thus, the first constraint in the above ILP forces every journey from $s$ to $d$ to be disrupted by at least one of the selected $\delta$-removals, such that $d$ is not reachable from $s$.

The above formulation is concise but has an exponential number of constraints because the number of possible journeys is exponential in the number of contacts. There also exists a compact ILP formulation which is less intuitive and omitted here for brevity. The complexity and the algorithm for solving the above ILP will be further discussed in Section \ref{com_mincut}.

\vspace{1mm}

\noindent \emph{\textbf{2) Survivability Metric: $\mathlarger{\mathsf{MaxFlow_{\delta}}}$}}

\vspace{1mm}

The second survivability metric, namely $\tt \mathsf{MaxFlow_{\delta}}$, is analogous to MaxFlow in static networks. Before the detailed definition of this metric, we first introduce the notion of $\delta$-disjoint journeys.

\begin{definition}[{$\delta$-disjoint Journey}]
A set of journeys from the source to the destination are $\delta$-disjoint if any two of these journeys do not use the same edge within $\delta$ time slots.
\end{definition}

\noindent Mathematically, suppose $\mathcal{J}$ is a set of $\delta$-disjoint journeys. For any two journeys $J_1,J_2\in \mathcal{J}$, if edge $e$ is used by $J_1$ in slot $t$, then $J_2$ cannot use the same edge $e$ from slot $t-\delta+1$ to slot $t+\delta-1$. In other words, sliding a window of $\delta$ slots over time, we can observe at most one active journey over each edge within the window. Figure \ref{disjoint_example} gives an example of $\delta$-disjoint journeys for the cases where $\delta=1$ and $\delta=2$. %Intuitively, each one of the $\delta$-disjoint journeys keeps a ``temporal distance" (or guardtime) of $\delta$ time slots from others.

\begin{figure}[ht]
\begin{center}
\includegraphics[width=3.2in]{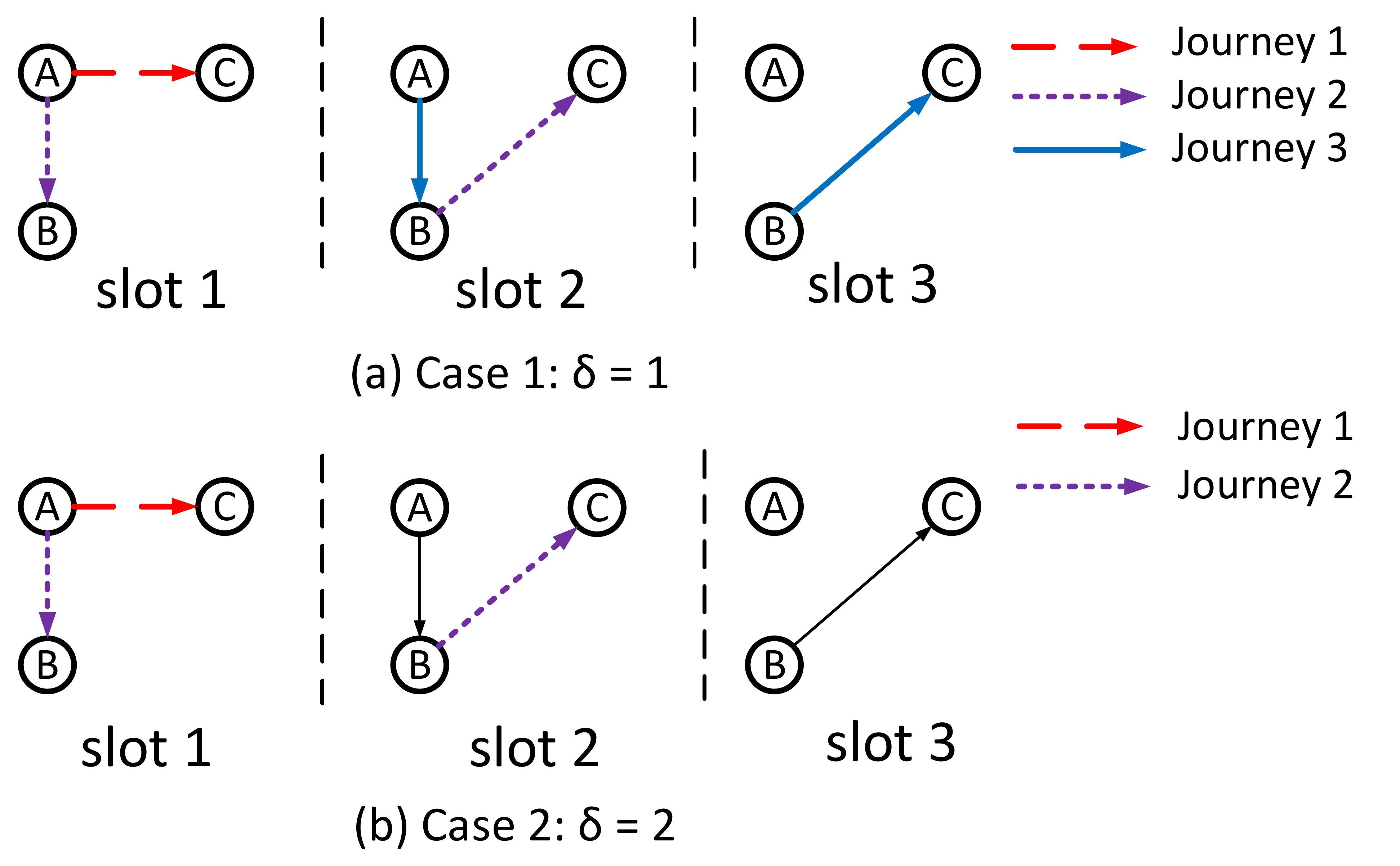}
\caption{Illustration of $\delta$-disjoint journeys. The source-destination pair is $\mathsf{(A,C)}$. (a) When $\delta=1$, any two different $\delta$-disjoint journeys cannot use the same link within the same slot, and there are three $\delta$-disjoint journeys. (b) When $\delta=2$, only two $\delta$-disjoint journeys exist since any link cannot be used by two $\delta$-disjoint journeys within 2 slots. For example, link $\mathsf{A\rightarrow B}$ has been used by Journey 2 in slot 1, so any other $\delta$-disjoint journey cannot use this link in slot 1 or 2.}
\label{disjoint_example}\vspace{-3mm}
\end{center}
\end{figure}

It is easy to see that each one of the $\delta$-disjoint journeys keeps a ``temporal distance" of $\delta$ slots from others. Due to the temporal distance, any failure that lasts for $\delta$ slots can influence at most one of these $\delta$-disjoint journeys. Consequently, the maximum number of $\delta$-disjoint journeys in a time-varying network is a good indicator of its survivability. The more $\delta$-disjoint journeys there exist, the more failures (lasting for $\delta$ slots) the network can survive. Now it is natural to introduce the second survivability metric $\mathsf{MaxFlow_{\delta}}$.

\begin{definition}[{$\mathsf{MaxFlow_{\delta}}$}]
$\mathsf{MaxFlow_{\delta}}$ is the maximum number of $\delta$-disjoint journeys from the source to the destination.
\end{definition}

\noindent \textbf{Discussion.} First, we would like to compare MaxFlow (for static networks) and $\mathsf{MaxFlow_{\delta}}$ (for time-varying networks). MaxFlow considers disjoint paths which require \emph{spatial disjointness}, i.e., any two disjoint paths never use the same link. This requirement is too demanding for time-varying networks because such networks often have sparse spatial connectivity. In the example of bus communication networks (see Section \ref{app1}), we will see that a time-varying network may not have any spatially-disjoint paths. Thus, MaxFlow is not an appropriate metric for time-varying networks. By comparison, $\mathsf{MaxFlow_{\delta}}$ considers $\delta$-disjoint journeys, which allows for \emph{temporal disjointness}.  Moreover, $\mathsf{MaxFlow_{\delta}}$  generalizes MaxFlow since we can simply set $\delta$ = $T$ so that $\delta$-disjoint journeys become spatially disjoint.

Second, $\mathsf{MaxFlow_{\delta}}$ not only gives us a measure of network survivability but also tells us how to achieve such survivability. The idea is similar to Disjoint-Path Protection in static networks \cite{dis_protect1}\cite{dis_protect2}, where disjoint paths are used as backup routes. In time-varying networks, we can send packets along different $\delta$-disjoint journeys  to increase transmission reliability. If we use $n$ $\delta$-disjoint journeys (i.e., $\mathsf{MaxFlow}_{\delta}\ge n$), the transmission can survive any $n-1$ failures that last for $\delta$ slots and is thus $(n-1,\delta)$-survivable.

\vspace{1mm}

\noindent \textbf{Formulation.} $\mathsf{MaxFlow_{\delta}}$ corresponds to the following ILP:
\[
\begin{split}
\max~~~~~~~~&\sum_{J\in \mathcal{J}_{sd}}x_{J}\\
\text{s.t.}~~~~~~~~&\sum_{J: (e,t)\in R(\delta,J)}x_{J}\le 1,~\forall (e,t)\in C\\
               &x_{J}\in \{0,1\},~\forall J\in \mathcal{J}_{sd}.
\end{split}
\]
Here, $x_J$ is a binary variable indicating whether journey $J$ should be added to the set of $\delta$-disjoint journeys. All the other notations have the same meanings as in the formulation of $\mathsf{MinCut_{\delta}}$. The first constraint checks every edge and forces this edge to be used by at most one of the $\delta$-disjoint journeys in any time window of $\delta$ slots. The above formulation also has an exponential number of constraints. A compact formulation also exists but is omitted for brevity. The complexity and the algorithms for solving the above ILP will be further investigated in Section \ref{com_maxflow}.

\subsection{Analysis of Metrics}\label{ana_metrics}
Recall that in static networks, the well-known Menger's Theorem shows that MinCut equals to MaxFlow; due to this equivalence, we can compute MaxFlow and MinCut efficiently (e.g., the Ford-Fulkerson algorithm). Hence, it is necessary to study the fundamental relationship between $\mathsf{MinCut_{\delta}}$ and $\mathsf{MaxFlow_{\delta}}$, in order to gain insights into their computation. Let $\mathsf{MinCut^R_{\delta}}$ and $\mathsf{MaxFlow^R_{\delta}}$ be the LP relaxation for the ILP formulation of $\mathsf{MinCut_{\delta}}$ and $\mathsf{MaxFlow_{\delta}}$, respectively. It is easy to show that $\mathsf{MinCut^R_{\delta}}$  is the \emph{dual problem} of $\mathsf{MaxFlow^R_{\delta}}$. By strong duality and the properties of LP relaxation, we make the following observation:
\vspace{-2mm}
\[
\mathsf{MaxFlow_{\delta}\le MaxFlow^{R}_{\delta}=MinCut^{R}_{\delta}\le MinCut_{\delta}}.\vspace{-2mm}
\]
As a result, as long as Menger's Theorem holds in time-varying networks (i.e., $\mathsf{MaxFlow_{\delta}= MinCut_{\delta}}$), all of the four quantities will be equivalent, and we can simply compute $\mathsf{MaxFlow_{\delta}}$ and $\mathsf{MinCut_{\delta}}$ by solving their LP relaxations. Interestingly, the following theorem shows that Menger's Theorem only ``conditionally" holds in time-varying networks.

\begin{theorem}\label{delta}
Time-varying graphs have the following survivability properties:\\
(I) If $\delta=1$, then Menger's Theorem holds for any time-varying graph, i.e., $\mathsf{MaxFlow_{1}= MinCut_{1}}$.\\
(II) For any $\delta\ge 2$, there exist instances of time-varying graphs such that $\mathsf{MaxFlow_{\delta}< {MinCut_{\delta}}}$. Moreover, the gap ratio $\frac{\mathsf{MinCut_{\delta}}}{\mathsf{MaxFlow_{\delta}}}$ can grow without bound.
\end{theorem}
\begin{proof}
See Appendix \ref{proof-delta}.
\end{proof}
Theorem \ref{delta} shows that Menger's Theorem could break down in time-varying graphs, which highlights a key difference between time-varying and static graphs. Due to this fundamental difference, the traditional techniques used to compute MaxFlow or MinCut in static networks, such as the Ford-Fulkerson algorithm, cannot be applied to time-varying graphs to compute $\mathsf{MaxFlow_{\delta}}$ or $\mathsf{MinCut_{\delta}}$. In the next section, we will further discuss the computation of the two metrics.

\section{Computational Issues}\label{comp}
In this section, we study the computational complexity and related algorithms for computing $\mathsf{MaxFlow_{\delta}}$ and $\mathsf{MinCut_{\delta}}$ in time-varying networks. %In particular, we will show that computing both metrics is NP-hard in a time-varying network and develop efficient approximation algorithms.
\subsection{Computation of { $\mathlarger{\mathsf{MaxFlow_{\delta}}}$}}\label{com_maxflow}
We start with the computation of  $\mathsf{MaxFlow_{\delta}}$ for an \emph{arbitrary value} of $\delta$, referred to as the $\boldsymbol{\delta}$-\textbf{MAXFLOW} problem.
The following theorem shows that this problem is even NP-hard to approximate.
\begin{theorem}\label{max_hardness}
$\delta$-MAXFLOW is NP-hard. It is even NP-hard to achieve $O(\sqrt{|E|})$-approximation, and this bound is tight.
\end{theorem}
\begin{proof}
See Appendix \ref{proof_max_hardness}.
\end{proof}
\noindent Note that to prove the tightness of the inapproximability bound, we just need to find an algorithm that achieves $O(\sqrt{|E|})$-approximation, which will be demonstrated later.

Next, we propose an approximation algorithm that attains the approximation lower bound in Theorem \ref{max_hardness}. Before we move on to the detailed algorithm description, it is necessary to introduce a short-hand term called \emph{interfering contact}.

\begin{definition}[{Interfering Contact}]
Consider a journey $J$. A contact $(e,t)$ is said to be an interfering contact of journey $J$ if there exists a contact $(e,t')$ used by $J$ such that $|t-t'|<\delta$.
\end{definition}

\noindent If $J$ is one of the $\delta$-disjoint journeys, then its interfering contacts cannot be used by any other $\delta$-disjoint journey.

Now we are ready to present a greedy algorithm for $\delta$-MAXFLOW, shown as Algorithm \ref{alg_maxflow1}. It first computes the Line Graph (see Section \ref{line}) of the original time-varying graph and then finds an $s$-$d$ path with the least number of nodes %\footnote{For example, simple Breadth-First Search (BFS) can find such a path in $O(|E'|)$ time, where $|E'|$ is the number of edges in the Line Graph.}
 in the Line Graph. By the property of Line Graphs (see Observation \ref{line_corre} in Section \ref{line}) , this path corresponds to a journey in the original time-varying graph; then we add this journey to the set of $\delta$-disjoint journeys. The next operation is to remove all the interfering contacts of this journey from the time-varying graph and reconstruct the Line Graph from the \emph{remaining time-varying graph}. If $s$ and $d$ are still connected in the Line Graph, the above procedure is repeated until $s$ and $d$ are disconnected. From the definition of interfering contacts, we can easily  verify that the obtained journeys are $\delta$-disjoint.
\begin{algorithm}[ht]
 \caption{Greedy Algorithm for $\delta$-MAXFLOW}\label{alg_maxflow1}
    \begin{algorithmic}[1]
\REQUIRE ~~\\

$\mathcal{G}$: the time-varying graph;\\
$(s,d)$: the source-destination pair;\\
$\delta$: the degree of temporal disjointness;\\

\ENSURE ~~\\

$J_1,\cdots,J_m$: a set of $\delta$-disjoint journeys.

\vspace{1mm}

\STATE Initialize $m=0$;

\STATE Compute the Line Graph of $\mathcal{G}$;\label{line_step}
\IF{$s$ and $d$ is disconnected in the Line Graph}
\STATE Go to step \ref{end};
\ENDIF \label{a3}
\STATE $m\leftarrow m+1$;
\STATE In the Line Graph, find an $s-d$ path $P_m$ that passes the least number of nodes (the corresponding journey is denoted by $J_m$);\label{a1}
\STATE Remove all the interfering contacts of $J_m$ from $\mathcal{G}$;\label{a2}
\STATE Go to step \ref{line_step};
\STATE END.\label{end}
\end{algorithmic}
\end{algorithm}

%Now we estimate the time complexity of this greedy algorithm. In each iteration (steps \ref{line_step}-\ref{a2}), we need to compute the Line Graph and the path with the least number of nodes. Recall that we denote $|C|$ the total number of contacts in the time-varying graph. Then it takes $O(|C|^2)$ time to construct the Line Graph and $O(|C|^2)$ time to compute the path with the least number of nodes (suppose BFS is used). Also note that the total number of iterations is at most $|C|$ since the number of $\delta$-disjoint journeys cannot exceed $|C|$ and each iteration adds one $\delta$-disjoint journey. Consequently, the overall time complexity of the greedy algorithm is $O(|C|^3)$.

Now we estimate the time complexity of this greedy algorithm. In each iteration (steps \ref{line_step}-\ref{a2}), we need to compute the Line Graph and the path with the least number of nodes. Recall that we denote $|C|$ the total number of contacts in the time-varying graph. Then it takes $O(|C|^2)$ time to construct the Line Graph and $O(|C|^2)$ time to compute the path with the least number of nodes (suppose BFS is used). Also note that the total number of iterations is at most $|C|$ since the number of $\delta$-disjoint journeys cannot exceed $|C|$ and each iteration adds one $\delta$-disjoint journey. Consequently, the overall time complexity of the greedy algorithm is $O(|C|^3)$.

The approximation ratio of this greedy algorithm is given in the following theorem.

\begin{theorem}\label{ratio_maxflow}
The greedy algorithm  attains $O(\sqrt{|E|})$ approximation for $\delta$-MAXFLOW, i.e., $\frac{\mathsf{OPT}}{\mathsf{ALG}}=O(\sqrt{|E|})$.
\end{theorem}
\begin{proof}
See Appendix \ref{proof_ratio_maxflow}.
\end{proof}

\noindent Clearly, the above approximation ratio attains the lower bound in Theorem \ref{max_hardness}. As a result, the greedy algorithm is the \textbf{optimal approximation algorithm} that achieves the best approximation ratio, and the inapproximability bound in Theorem \ref{max_hardness} is tight. In practice, the greedy algorithm also performs extremely well, as is demonstrated by the following numerical results.

\vspace{1mm}

\noindent \textbf{Numerical Results for the Greedy Algorithm.} In order to understand the performance of the greedy algorithm, we compare it with the optimal solution to $\delta$-MAXFLOW. In our experiment, 1000 random time-varying graphs are tested. Each network has 20 nodes and the underlying static graph is a random scale-free graph. The time horizon is $T=20$ slots and we assume each link is active with a probability $p=0.5$ in each slot. The source-destination pair is also randomly selected. The optimal solution to $\delta$-MAXFLOW is derived by directly solving its ILP formulation. Figure \ref{gap_maxflow} shows the comparison, where the approximation gap is calculated by $\frac{\mathsf{OPT}-\mathsf{ALG}}{\mathsf{ALG}}$. We can observe that the approximation gap is usually less than 8\%, much better than the theoretical bound in Theorem \ref{ratio_maxflow}.

\begin{figure}[ht]
\begin{center}
\includegraphics[width=2.7in]{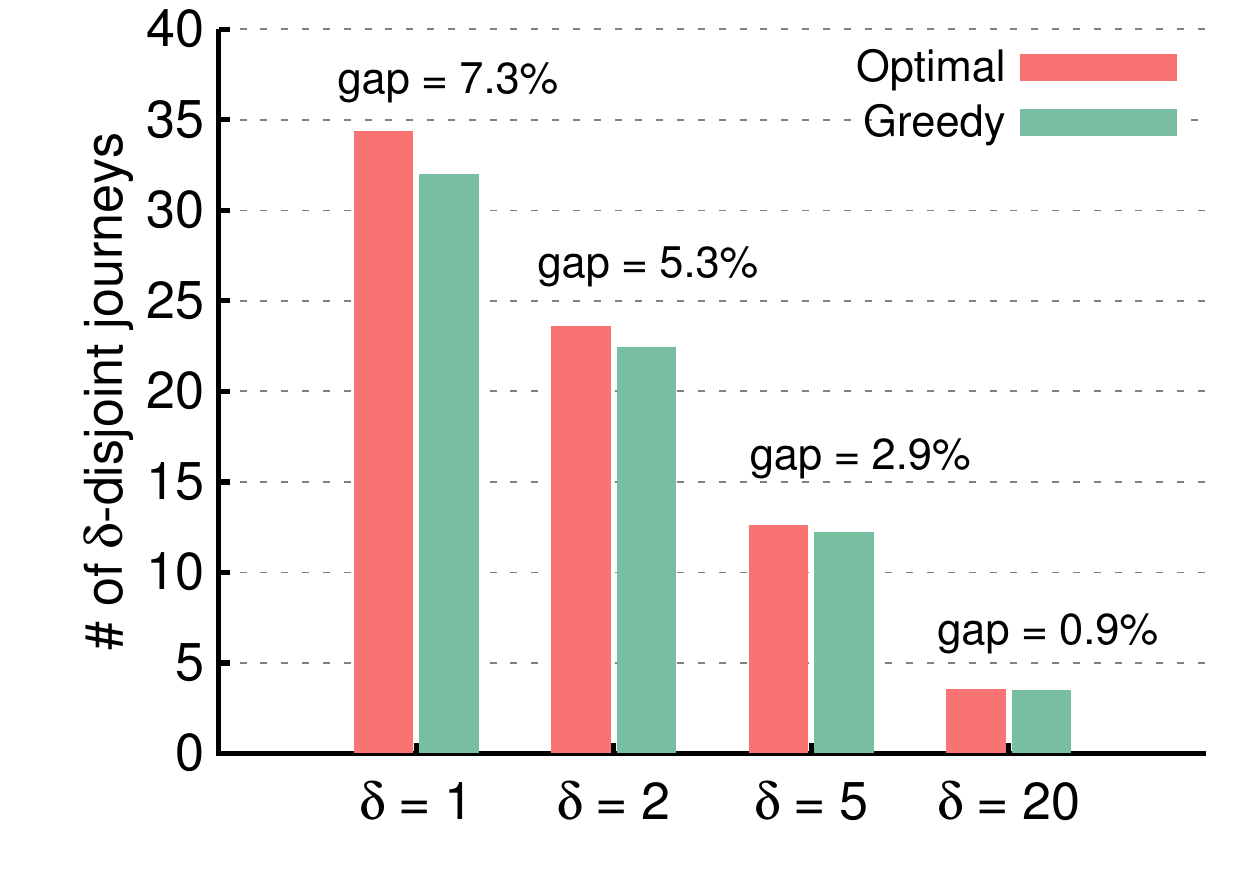}
\caption{Comparison between the greedy algorithm (Algorithm \ref{alg_maxflow1}) and the optimal solution to $\delta$-MAXFLOW.}
\label{gap_maxflow}\vspace{-3mm}
\end{center}
\end{figure}

\subsection{Computation of {$\mathlarger{\mathsf{MinCut_{\delta}}}$}}\label{com_mincut}
In this section, we study the computation of $\mathsf{MinCut_{\delta}}$ for an arbitrary value of $\delta$, referred to as the $\boldsymbol{\delta}$-\textbf{MINCUT} problem.

\noindent The complexity of $\delta$-MINCUT is given in Theorem \ref{mincut_hardness}.

\begin{theorem}\label{mincut_hardness}
$\delta$-MINCUT is NP-hard.
\end{theorem}
\begin{proof}
Kempe \emph{et al.} \cite{con} showed that in a special type of time-varying graphs, where each link is active for only one slot, it is NP-hard to determine whether there exists a set of $k$ nodes whose permanent removals can disconnect the source-destination pair. This is obviously a restricted instance of the node-version $\delta$-MINCUT problem, which implies that the node-version $\delta$-MINCUT is NP-hard. Moreover, it can be verified that node-version problems are just a special case of edge-version problems by using node splitting (see \cite{LP}, Chapter 7.2). Hence, the edge-version $\delta$-MINCUT problem is also NP-hard.
\end{proof}

Due to the computational intractability of $\delta$-MINCUT, we present an approximation algorithm  (referred to as the \emph{min-weight algorithm}) for $\delta$-MINCUT. The algorithm proceeds in three steps.

\vspace{1mm}

\noindent $\bullet$ \textbf{Step 1:} Assign a weight to each contact according to its ``temporal closeness" to other contacts. Intuitively, if there are more contacts in the ``temporal neighborhood" of the given contact, then a $\delta$-removal (i.e., a $\delta$-slot failure) of this contact will disable more neighboring contacts at the same time. Hence, this contact should be given a smaller weight such that it has a higher priority of being removed. We let the weight of a contact be inversely proportional to the number of its ``neighboring" contacts (see \textsc{SetWeight} in Algorithm \ref{heu}).

\vspace{1mm}

\noindent $\bullet$ \textbf{Step 2:} Compute $\mathsf{MinCut}_{1}$ over the weighted time-varying graph.  Note that Property (I) in Theorem \ref{delta} still holds in weighted time-varying graphs, so $\mathsf{MinCut}_1$ can be efficiently computed (e.g., by solving the LP relaxation). After this step, we obtain a set of contacts $S^*$ with the smallest sum of weights whose removals will disconnect the source-destination pair.

\vspace{1mm}

\noindent $\bullet$ \textbf{Step 3:} Compute the $\delta$-cover of $S^*$, i.e., the \emph{smallest} set of $\delta$-removals needed to cover all the contacts in $S^*$. For example, suppose $S^*=\{(e_1,1),(e_1,2),(e_2,2),(e_2,4)\}$ and $\delta=2$. Then we need at least three $\delta$-removals to cover all the contacts in $S^*$: one for $(e_1,1)$ and $(e_1,2)$, one for $(e_2,2)$ and one for $(e_2,4)$; this means that $|\mathsf{Cover_{\delta}}(S^*)|=3$. Finally, the $\delta$-cover of $S^*$ is returned as a feasible solution to $\delta$-MINCUT.

\begin{algorithm}[ht]
 \caption{Min-Weight Algorithm for $\delta$-MINCUT}\label{heu}
    \begin{algorithmic}[1]
\STATE Call \textsc{SetWeight} to compute the weight for each contact;%
\STATE Compute $\mathsf{MinCut}_1$ over the weighted time-varying graph, where we obtain a set of contacts $S^*$ with the smallest sum of weights whose removals will disconnect the source-destination pair;
\STATE Return the $\delta$-cover of $S^*$ as the solution.
\STATE \textbf{Procedure:} \textsc{SetWeight}
\FOR {each contact $(e,t)$}
\STATE Scan all the $\delta$-slot windows containing $(e,t)$, and find the one that contains the maximum number of contacts (say containing $K_{e,t}$ contacts);
\STATE Set $\omega_{e,t}=\frac{1}{K_{e,t}}$;
\ENDFOR
\end{algorithmic}
\end{algorithm}

The performance of the above min-weight algorithm is given in the following theorem.
\begin{theorem}\label{heu_bound}
The min-weight algorithm (Algorithm \ref{heu}) achieves $\delta$-approximation for $\delta$-MINCUT, i.e., $\frac{\mathsf{ALG}}{\mathsf{OPT}}\le \delta$.
\end{theorem}
\begin{proof}
See Appendix \ref{proof_heu_bound}.
\end{proof}

\noindent \textbf{Numerical Results for the Min-Weight Algorithm.} The simulation setting is the same as that used for Algorithm \ref{alg_maxflow1}. Figure \ref{gap_mincut} shows the comparison between the min-weight algorithm (Algorithm \ref{heu}) and the optimal solution to $\delta$-MINCUT.
We notice that the min-weight algorithm is close to the optimum: the approximation gap\footnote{The approximation gap is calculated by $\frac{\mathsf{ALG}-\mathsf{OPT}}{\mathsf{OPT}}$.} is less than 10\% for a relatively small value of $\delta$; in particular, the approximation gap is zero when $\delta=1$. The final observation is that the approximation gap becomes larger with the increase in $\delta$; this tendency is consistent with the theoretical approximation ratio of $\delta$.

\begin{figure}[ht]
\begin{center}
\includegraphics[width=2.7in]{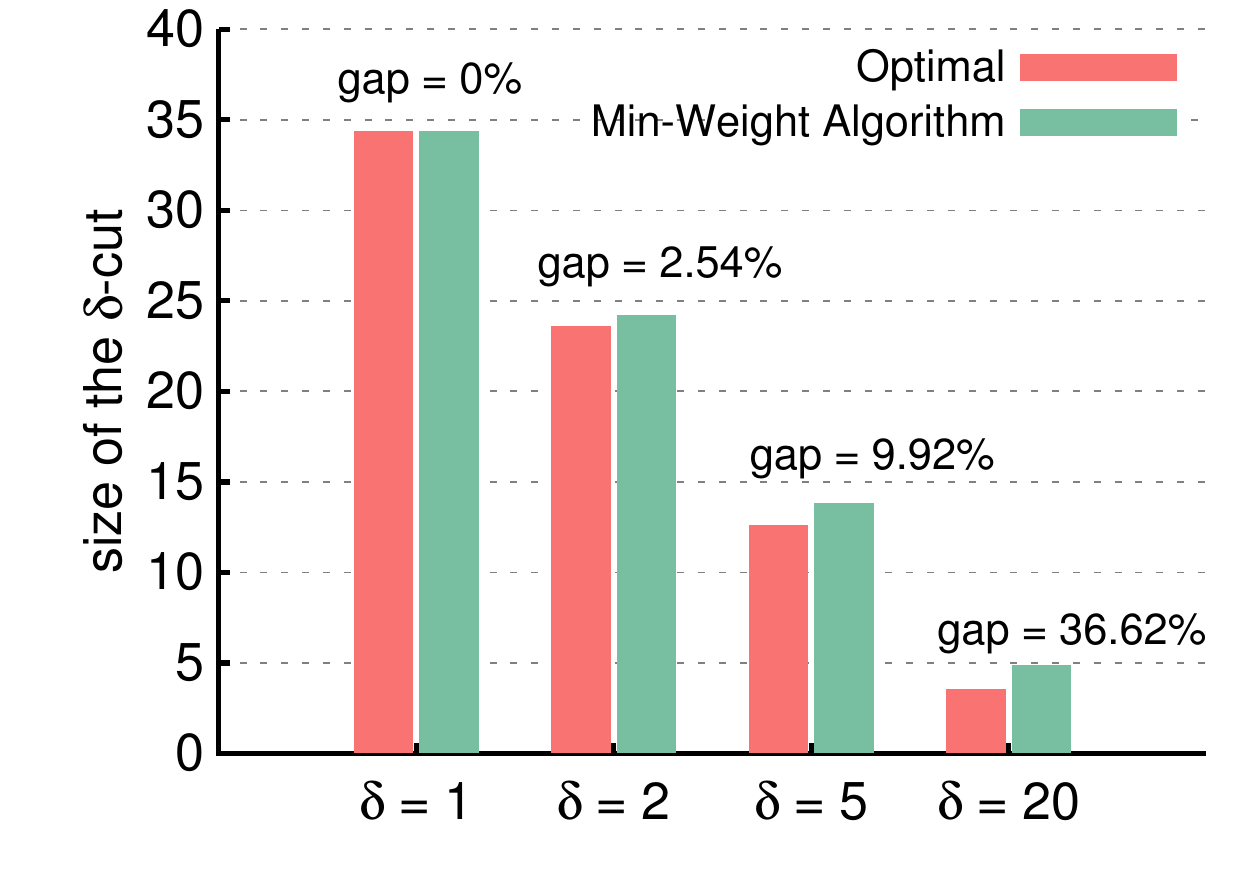}
\caption{Comparison between the min-weight algorithm (Algorithm \ref{heu}) and the optimal result to $\delta$-MINCUT.}
\label{gap_mincut}\vspace{-3mm}
\end{center}
\end{figure}

\section{Application: Bus Communication Networks}\label{app1}
In this section, we demonstrate how to use our survivability framework to facilitate the design of robust networks in practice. To be more specific, we exploit  $\delta$-disjoint journeys to design a \emph{survivable routing} protocol for a real-world bus communication network \cite{trace}. Each bus in the network has a pre-designed route and is equipped with an 802.11 radio that constantly scans for other buses. Since the route of each bus is designed in advance, we can make a \emph{coarse prediction} about bus mobility and the evolution of their communication topology. As a result, we can convert this bus communication network into a time-varying graph whose topology changes according to the estimated bus mobility. However, the prediction may not be perfect due to various reasons such as unexpected obstacles, traffic accidents, traffic jam, etc.  The goal of survivable routing is to reduce the packet loss rate due to these unpredictable failures.

In the rest of this section, we first present the design of the survivable routing protocol using $\delta$-disjoint journeys. Then we discuss trace statistics, simulation settings and results.
\subsection{Survivable Routing Protocol: DJR}\label{proto}
The basic idea of this protocol is to replicate each packet at the source and send these copies along multiple $\delta$-disjoint journeys obtained by solving $\delta$-MAXFLOW. When at least one of these copies reaches the destination, the original packet is successfully delivered. This replication-based protocol is referred to as Disjoint-Journey Routing (DJR). The advantages of DJR over other reliable routing protocols are as follows.

\vspace{1mm}

\noindent $\bullet$ {\emph{Simplicity of Deployment in Time-varying Networks.} Static networks usually deploy ARQ at the data link layer and TCP at the transport layer for error recovery. However, due to the lack of connectivity, it is not only difficult to get timely ACK at the sender but also hard to find opportunities for retransmissions. In contrast, DJR does not require any feedback, which greatly simplifies the data link layer and the transport layer (no need for error recovery). In addition, as a network-layer protocol, it can be  combined with FEC codes at the physical layer (e.g., erasure code \cite{redundancy}) to achieve a better performance.}

\vspace{1mm}

\noindent $\bullet$ {\emph{Temporal Diversity.} Traditional survivable routing protocols rely on spatial diversity, such as Disjoint-Path Routing (DPR)  \cite{dis_protect1}\cite{dis_protect2}, where spatially-disjoint paths are used to recover packets. However, spatial diversity is a demanding requirement in networks with sparse and intermittent connectivity. We will demonstrate that it is hard to find even two spatially-disjoint paths in the bus network. By comparison, DJR exploits temporal diversity to combat failures and is well suited for time-varying networks, especially when failures are transient.}

\vspace{1mm}

\noindent $\bullet$ {\emph{Two-dimensional Tunability.} Our survivability framework has two natural parameters, namely $n$ and $\delta$. Hence, the tunability of DJR is also in two dimensions: we can both tune the number of $\delta$-disjoint journeys to use, and also adapt the degree of temporal disjointness. By comparison, existing survivable routing protocols (e.g., \cite{DTN_routing1, DTN_routing2, DTN_routing3}) lack such flexibility. }
\subsection{Traces}\label{tr_stat}
We use the trace from UMassDieselNet \cite{trace} where a public bus transportation system was operated around Amherst, Massachusetts. The trace records the contacts among 21 buses in 9 days, which roughly reflects bus mobility over the pre-designed bus routes. We use such contact information as a coarse prediction for the states of bus-to-bus links in the 9-day period. However, we assume that the prediction is imperfect and unpredictable failures may disable these contacts (the failure model will be introduced in the next section).

To facilitate our subsequent discussion, we pre-process the raw trace and observe two important features of this bus communication network. The first observation is the ``bursty" structure of contacts between any two buses; that is, buses only communicate with each other occasionally. Figure \ref{trace_stat}(a) illustrates such a bursty structure for a typical pair of buses. The second observation is that most connections in this network last for only a short period of time. As is shown in Figure \ref{trace_stat}(b), most contacts span less than 20s.

\begin{figure}[t]
\begin{center}
\includegraphics[width=3.3in]{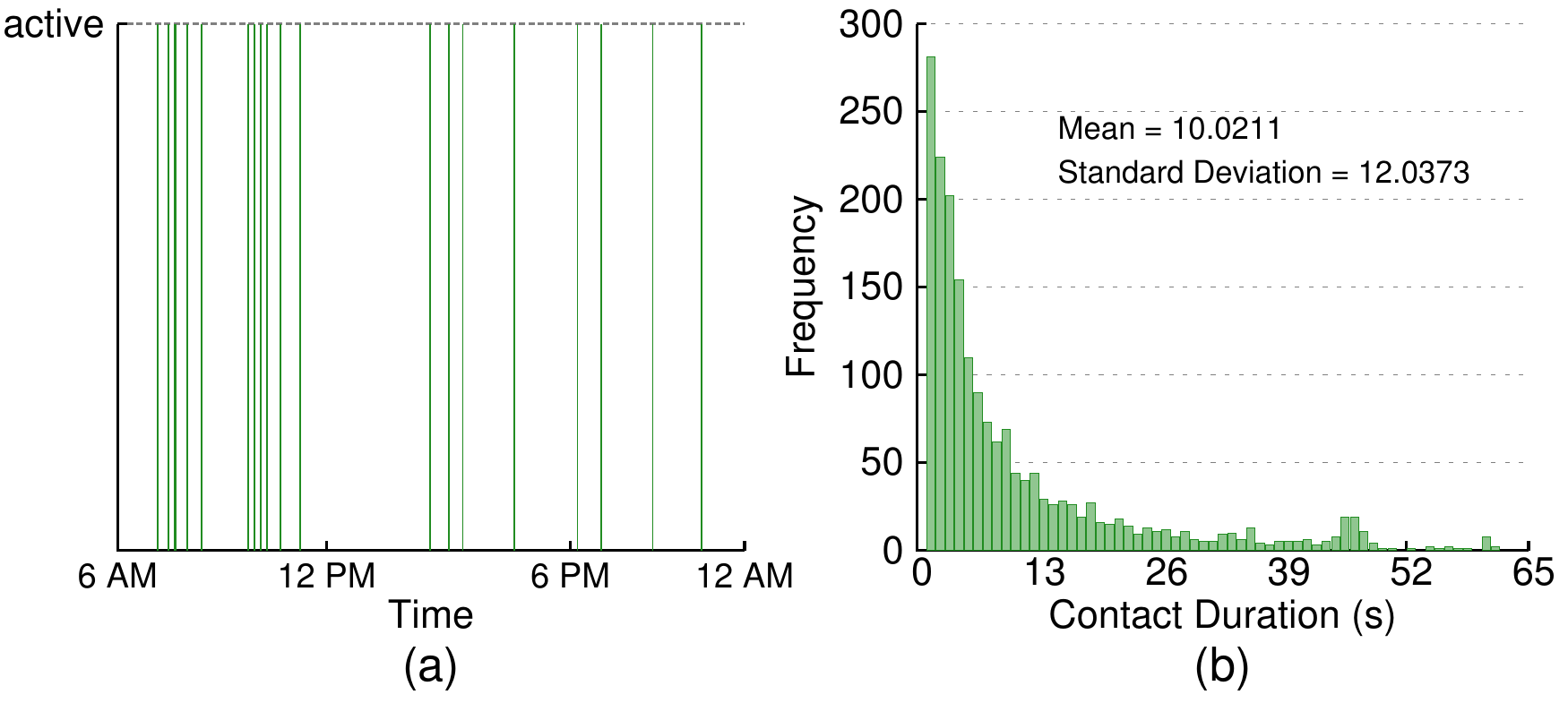}
\caption{Statistical structures of the bus communication network. (a) The bursty pattern of the contacts between a typical pair of buses. (b) Histogram for contact durations. Most contacts only last for a short period of time.}
\label{trace_stat}\vspace{-3mm}
\end{center}
\end{figure}

%\begin{figure*}[htbp]
%\centering
%\begin{minipage}[t]{0.32\textwidth}
%\centering
%\includegraphics[width=2.1in]{Figures/bus1}\vspace{-3mm}
%\caption{The total number of $\delta$-disjoint journeys in the bus communication network.}
%\label{bus1}
%\end{minipage}%
%\hspace{0.25cm}
%\begin{minipage}[t]{0.32\textwidth}
%\centering
%\includegraphics[width=2.1in]{Figures/bus2}\vspace{-3mm}
%\caption{Influence of $n$ and $\delta$ on packet loss rates (DDL=300s, $p$=0.05, $d$=60s). }
%\label{bus2}
%\end{minipage}%
%\hspace{0.25cm}
%\begin{minipage}[t]{0.32\textwidth}
%\centering
%\includegraphics[width=2.3in]{Figures/bus3}\vspace{-3mm}
%\caption{Comparison between DJR and DPR (DDL = 300s, $p$=0.05, $n$=3). }
%\label{bus3}
%\end{minipage}
%\end{figure*}

\subsection{Simulation Settings}
In our simulation, the slot length is identical to the trace resolution, i.e., one second. According to the measurement in \cite{trace}, the average transmission rate is about 1.64Mbps. If the packet size is set to be 1KB, the transmission time of one packet is nearly negligible as compared to the slot length, which implies zero link-traversal delay. Each packet has a deadline (DDL) after which it will be dropped from the network; naturally, the packet deadline can be modeled by the time horizon $T$ of the corresponding time-varying graph.  A packets is generated between a random source-destination pair immediately after the previous packet expires or gets delivered. In addition, at most $n$ copies are allowed, meaning that we can use at most $n$ $\delta$-disjoint journeys to send these copies.  Algorithm \ref{alg_maxflow1} is used to compute $\delta$-disjoint journeys.

Since it is impossible to precisely predict future topology changes, we impose random failures on the time-varying graph generated from the trace. For each link, we let failures occur in each slot with a certain probability $p$, and the duration of each failure is uniformly distributed within $[0,d]$ seconds. The performance metric is the packet loss rate, i.e., the fraction of packets that fail to reach the destination before the deadline.
%Note that if $d\ge T$, failures are permanent (transient otherwise).
\subsection{Total Number of $\mathlarger{\delta}$-Disjoint Journeys}
We first look at the maximum number of $\delta$-disjoint journeys in the bus communication network (Figure \ref{bus1}). First, it can be observed that there exist very few $\delta$-disjoint journeys in this network: less than three $\delta$-disjoint journeys when $\delta\ge 5$. Particularly, only one $\delta$-disjoint journey exists when $\delta$ is relatively large, which means that it is almost impossible to find even two journeys that are spatially disjoint (i.e., $\delta=T$). This observation indicates the lack of spatial connectivity in this bus network and implies the inefficiency of traditional Disjoint-Path Routing in networks with intermittent connectivity since such a protocol only relies on spatial diversity. Second, we can observe the diminishing return for the number of $\delta$-disjoint journeys: beyond a certain value of $\delta$, the increase of $\delta$ no longer reduces the number of $\delta$-disjoint journeys. Such a tendency is due to the bursty contact structure in this network (see Section \ref{tr_stat}). %For example, if a contact only lasts for 20s while the next contact occurs 1 hour later, the temporal distance between any two journeys that use this link is at most $20$s within the deadline (unless the packet deadline is greater than 1 hour), and any larger $\delta$ actually prevents any two $\delta$-disjoint journeys from using the same link within the deadline, which essentially becomes spatial disjointness.
The final observation is that extending the packet deadline increases the total number of $\delta$-disjoint journeys since there are more transmission opportunities within a longer deadline.

\begin{figure}[t]
\begin{center}
\includegraphics[width=2.5in]{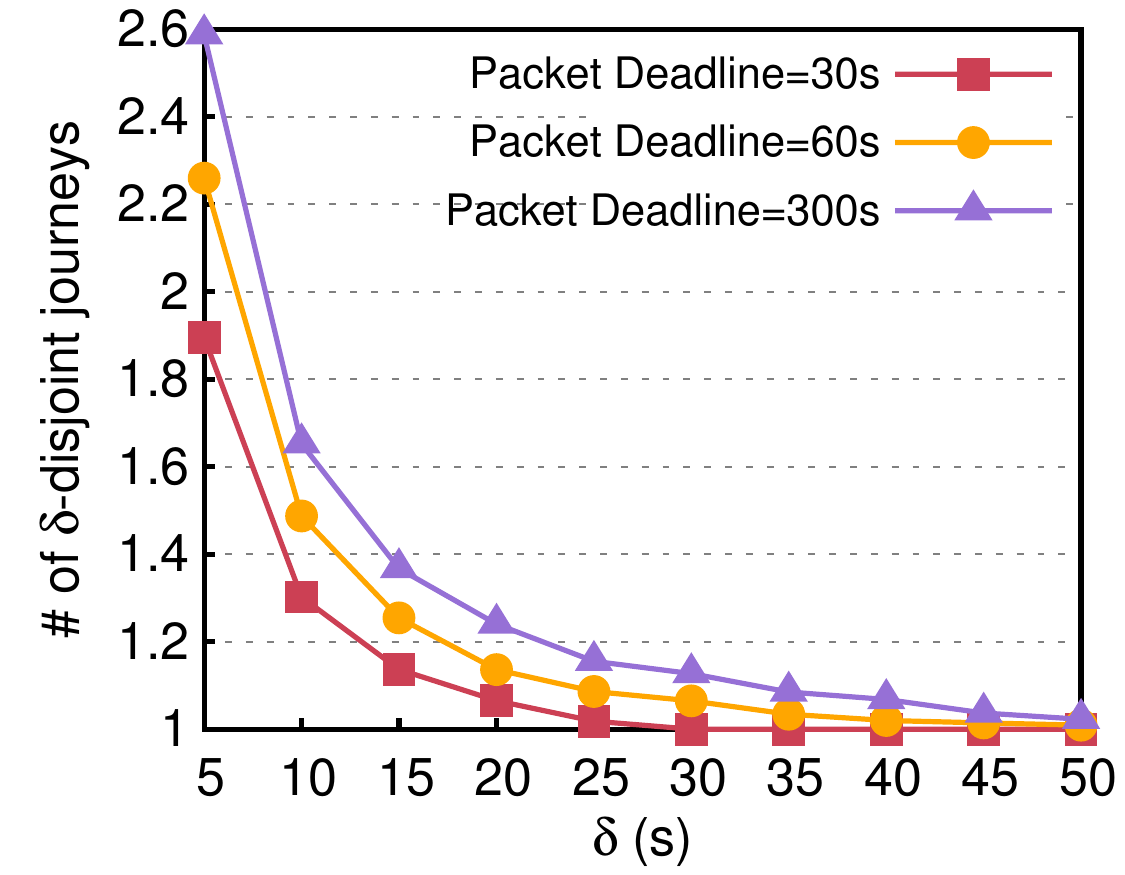}
\caption{The total number of $\delta$-disjoint journeys in the network.}
\label{bus1}\vspace{-3mm}
\end{center}
\end{figure}

\subsection{Tunability of DJR}
Next, we study the two-dimensional tunability of DJR (Figure \ref{bus2}). We first investigate the tunability of $n$, i.e., the maximum number of copies we are allowed to produce or the maximum number of $\delta$-disjoint journeys we can use. If we are  allowed to use only one of the $\delta$-disjoint journeys ($n=1$), DJR is ineffective and the packet loss rate remains at a high level regardless of the value of $\delta$. If we can use more $\delta$-disjoint journeys, the packet loss rate is significantly reduced (of course, more redundant copies are created).

The influence of $\delta$ is more interesting. With the increase of $\delta$, the packet loss rate first goes down and then increases; this tendency can be explained as follows. When $\delta$ is small, there exist many  $\delta$-disjoint journeys and we can choose any $n$ of them to transmit copies of packets. With a fixed number of disjoint journeys, it is known that larger temporal disjointness makes the network more robust since it can survive failures of longer duration. Hence, the packet loss rate first goes down. However, the increase of $\delta$ also leads to the reduction in the number of $\delta$-disjoint journeys (see Figure \ref{bus1}); beyond a certain value of $\delta$, the number of $\delta$-disjoint journeys becomes smaller than $n$ and we have to send copies over fewer than $n$ disjoint journeys, which means that the network can survive fewer failures. Therefore, although temporal disjointness continues to grow, the reduction in the  number of available disjoint journeys makes the loss rate increase. Moreover, we can observe that there exists an ``optimal" value of $\delta$ which minimizes the packet loss rate (highlighted by shaded circles). In fact, this optimal value is the maximum $\delta$ such that $\mathsf{MaxFlow}_{\delta} \ge n$.

\begin{figure}[t]
\begin{center}
\includegraphics[width=2.3in]{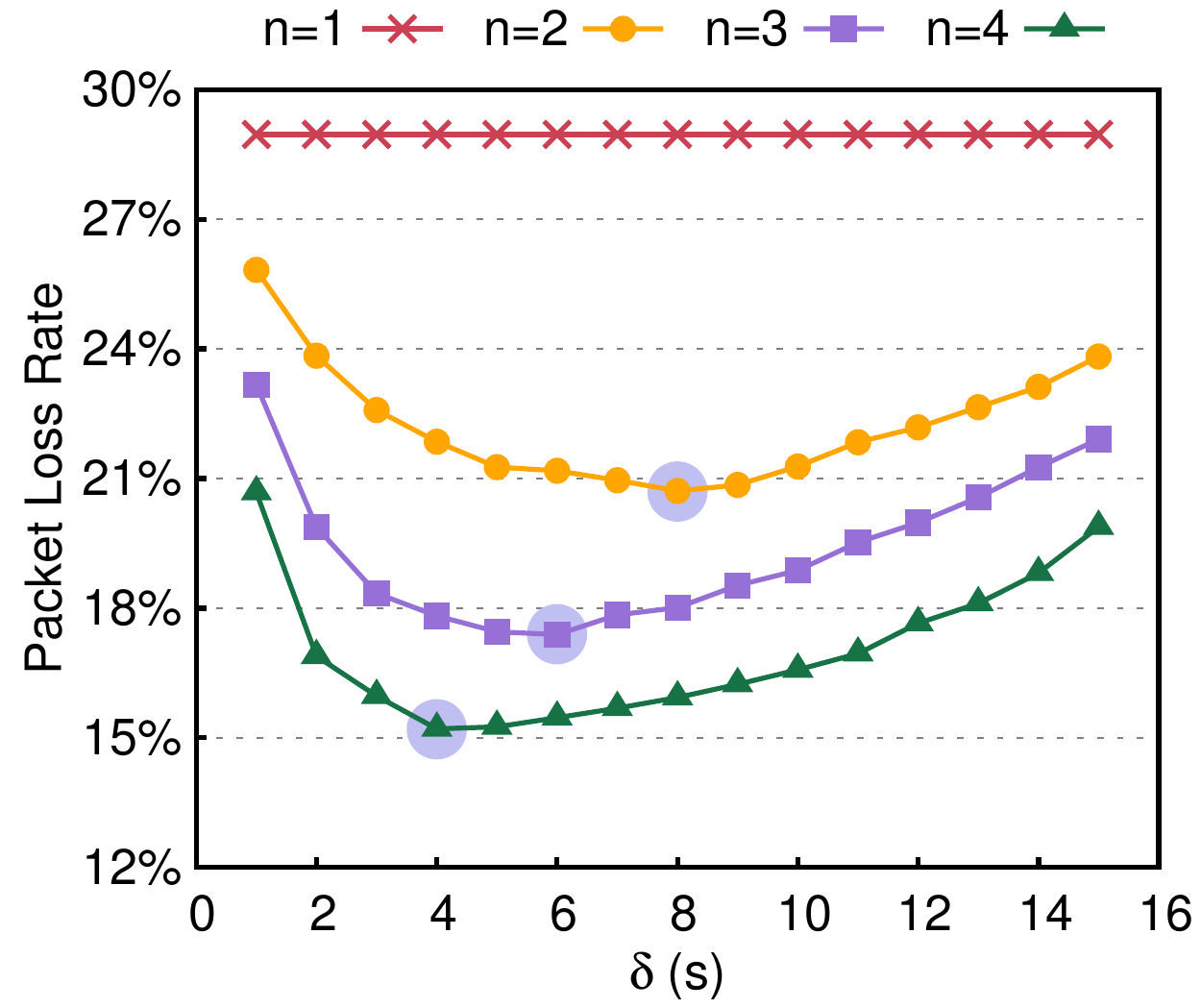}
\caption{Influence of $n$ and $\delta$ on packet loss rates (DDL=300s, $p$=0.05, $d$=60s).}
\label{bus2}\vspace{-3mm}
\end{center}
\end{figure}

\section{Related Work}\label{related_works}
\noindent \textbf{Time-varying Graphs.}
%With the advent of many dynamic networks (e.g., mobile networks), time-varying graphs have been recognized as a powerful modeling tool.
There is extensive literature seeking to define metrics for time-varying graphs, such as connectivity \cite{TVG1, con, TVG11}, distance \cite{TVG3}, centrality \cite{TVG6, TVG7}, diameter \cite{TVG8,TVG9}, etc. The combinatorial properties of time-varying graphs are also an active research area. For example, Kranakis \emph{et al.} focused on finding connected components in a time-varying graph; Ferreira \emph{et al.} investigated the complexity for computing the shortest journey \cite{TVG3} and the minimum spanning tree \cite{TVG11} (see the survey \cite{TVG2}).

\vspace{1mm}

\noindent \textbf{Survivability in Time-varying Networks.} Despite the extensive research on time-varying graphs, there is very little literature on survivability of time-varying networks. The closest work to ours was done by Berman \cite{vul} and Kleinberg \emph{et al.} \cite{con}. They discussed vulnerability in so-called ``edge-scheduled networks" or ``temporal networks" where each link is active for exactly one slot and only permanent failures happen. Our work considers a more general graph model while leveraging the temporal features of failures, thus generalizing their results. Scellato \emph{et al.} \cite{TVG4} investigated a similar problem in random time-varying graphs and proposed a metric called ``temporal robustness". By comparison, our framework is deterministic, thus guaranteeing the worst-case survivability. Li \emph{et al.} \cite{TVG-unreliable} studied a related but different problem in time-varying networks; specifically, they proposed heuristic algorithms to find the the min-cost subgraph of a probabilistic time-varying graph such that the probability that the subgraph is temporally connected exceeds a certain threshold.

\vspace{1mm}

\noindent \textbf{Time-varying Graphs and DTNs} An important application scenario of time-varying graphs is Delay Tolerant Networks (DTN), where nodes have intermittent connectivity and can only send packets opportunistically. The primary goal of DTN is  to improve the packet delivery ratio via some routing schemes, and there is extensive literature in this area, such as \cite{redundancy, DTN_routing1, DTN_routing2, DTN_routing3, TVG-unreliable}. In contrast, our work does not focus on any specific routing algorithm. Instead, this paper is intended to understand the inherent survivability properties of a time-varying network, which can facilitate the design of survivable routing algorithms in DTNs (e.g., Section \ref{app1}).
\section{Conclusions}\label{conclusion}
In this paper, we propose a new survivability framework for time-varying networks, namely $(n,\delta)$-survivability. In order to evaluate $(n,\delta)$-survivability, two metrics are proposed: $\mathsf{MinCut_{\delta}}$ and $\mathsf{MaxFlow_{\delta}}$.
%These metrics are analogous to MinCut and MaxFlow in static graphs but account for the temporal features of time-varying networks, such as failure duration.
We analyze the fundamental relationship between the two metrics and show that Menger's Theorem only conditionally holds in time-varying graphs. As a result, computing both survivability metrics is NP-hard. To resolve the computational intractability, we develop several approximation algorithms. Finally, we use trace-driven simulations to demonstrate the application of our framework in a real-world bus communication network.

\appendix
\section{Proofs}
\subsection{Proof to Theorem \ref{delta}} \label{proof-delta}
\subsubsection{Proof to Property (I)}
Consider a time-varying graph $\mathcal{G}$ with the source $s$ and the destination $d$.
Let $\mathsf{MaxFlow}$ be the maximum number of \emph{node-disjoint} paths from $s$ to $d$ in the Line Graph of $\mathcal{G}$ and $\mathsf{MinCut}$ be the cardinality of the smallest \emph{node cut} that separates $s$ and $d$ in the Line Graph. It is not hard to verify the following lemma.
\begin{lemma}\label{line_eq}
$\mathsf{MaxFlow_1=MaxFlow}$ and $\mathsf{MinCut_1=MinCut}$.
\end{lemma}

\noindent \textbf{Remark:} Lemma \ref{line_eq} does not holds for $\delta\ge 2$. For example, if $\delta=2$, there is only one $\delta$-disjoint journey in Figure \ref{line_example}(a) but there are two node-disjoint paths in its Line Graph.

\vspace{1mm}

Now we can apply the node-version Menger's Theorem to the Line Graph and obtain $\mathsf{MaxFlow=MinCut}$. By Lemma \ref{line_eq}, we can conclude that
\[
\mathsf{MaxFlow_1=MaxFlow=MinCut=MinCut_1}.
\]
\subsubsection{Proof to Property (II)}
The non-trivial part is to show that the gap ratio can be arbitrarily large. We construct a family of time-varying graphs $\{\mathcal{G}_k\}_{k\ge 1}$  such that  $\frac{\mathsf{MinCut_{\delta}}}{\mathsf{MaxFlow_{\delta}}}= k$ for any $\delta\ge 2$ in the $k$-th graph. The constructions for $k=1,2,3$ are shown in Figure \ref{bd}. We can observe that $\mathcal{G}_1$ is a single-level graph; $\mathcal{G}_2$ is built upon $\mathcal{G}_1$, where the first level is exactly $\mathcal{G}_1$; similarly, $\mathcal{G}_3$ is built upon $\mathcal{G}_2$, where the first two levels correspond to $\mathcal{G}_2$.

We use inductions to prove that $\mathsf{MaxFlow_{\delta}}=1$ while $\mathsf{MinCut_{\delta}}=k$ for any $\delta\ge 2$ in the $k$-th graph $\mathcal{G}_k$. For brevity, we only demonstrate the induction philosophy from  $\mathcal{G}_1$ to $\mathcal{G}_2$ while its generalization is easy.

\vspace{1mm}

\noindent $\bullet$ In $\mathcal{G}_1$, the source-destination pair is $(s,d_1)$. It is obvious that $\mathsf{MaxFlow_{\delta}}= \mathsf{MinCut_{\delta}}=1$ for any $\delta\ge 2$.

\vspace{1mm}

\noindent $\bullet$ In $\mathcal{G}_2$, the source-destination pair is $(s,d_2)$. We want to show that $\mathsf{MaxFlow_{\delta}}=1$ but $\mathsf{MinCut_{\delta}}=2$ for any $\delta\ge 2$. To see $\mathsf{MaxFlow_{\delta}}=1$, we notice that there are two possible choices for traveling from $s$ to $d_2$. One is via node $d_1$ and the other is to directly descend to level 2. The former choice yields only one $\delta$-disjoint journey from $s$ to $d_2$ since we know from $\mathcal{G}_1$ that there is only one $\delta$-disjoint journey from $s$ to $d_1$. For the latter choice, the only possibility is $s\rightarrow v_{2,1}\rightarrow v_{2,2}\rightarrow v_{2,3}\rightarrow d_2$ but this journey cannot be $\delta$-disjoint of any journey in the first choice (i.e., via node $d_1$) for any $\delta\ge 2$. Hence, there is only one $\delta$-disjoint journey from $s$ to $d_2$, i.e., $\mathsf{MaxFlow_{\delta}}=1$ for any $\delta\ge 2$. Now it remains to show $\mathsf{MinCut_{\delta}}=2$ and we prove this by showing that any single $\delta$-removal cannot disconnect $d_2$ from $s$. If the $\delta$-removal takes place in level 1, there exists a feasible journey from $s$ to $d_2$ via $s\rightarrow v_{2,1}\rightarrow v_{2,2}\rightarrow v_{2,3}\rightarrow d_2$ in slots 4, 5, 6, 7. If the $\delta$-removal occurs to some contact outside level 1, the journey from $s$ to $d_1$ is still available. Moreover, there exists at least one journey from $d_1$ to $d_2$ since there are two spatially disjoint journeys from $d_1$ to $d_2$ (one journey is via $d_1\rightarrow v_{2,1}\rightarrow v_{2,2}\rightarrow d_2$ in slots 3, 4, 5,  and the other journey is via $d_1\rightarrow v_{2,3}\rightarrow d_2$ in slots 3, 7). As a result, $d_2$ is still reachable from $s$ via $s\rightarrow d_1 \rightarrow d_2$. Now it is safe to conclude that any single $\delta$-removal cannot disconnect $d_2$ from $s$, which implies $\mathsf{MinCut_{\delta}} \ge 2$. Note that $d_2$ can be easily made unreachable from $s$ with 2 $\delta$-removals (e.g., disable the two contacts from $s$). Therefore, $\mathsf{MinCut_{\delta}} = 2$

\vspace{1mm}

Note that the key part in $\mathcal{G}_2$ is the ``shortcut edge" $v_{2,2}\rightarrow d_2$ which can only be used by journeys that travel through $d_1$. Following the similar line of induction (with minor modifications), we can show  $\mathsf{MaxFlow_{\delta}}=1$ and $\mathsf{MinCut_{\delta}}=k$ for any $\delta\ge 2$ in the $k$-th graph $\mathcal{G}_k$.

\begin{figure}[ht]
\begin{center}
\includegraphics[width=3.4in]{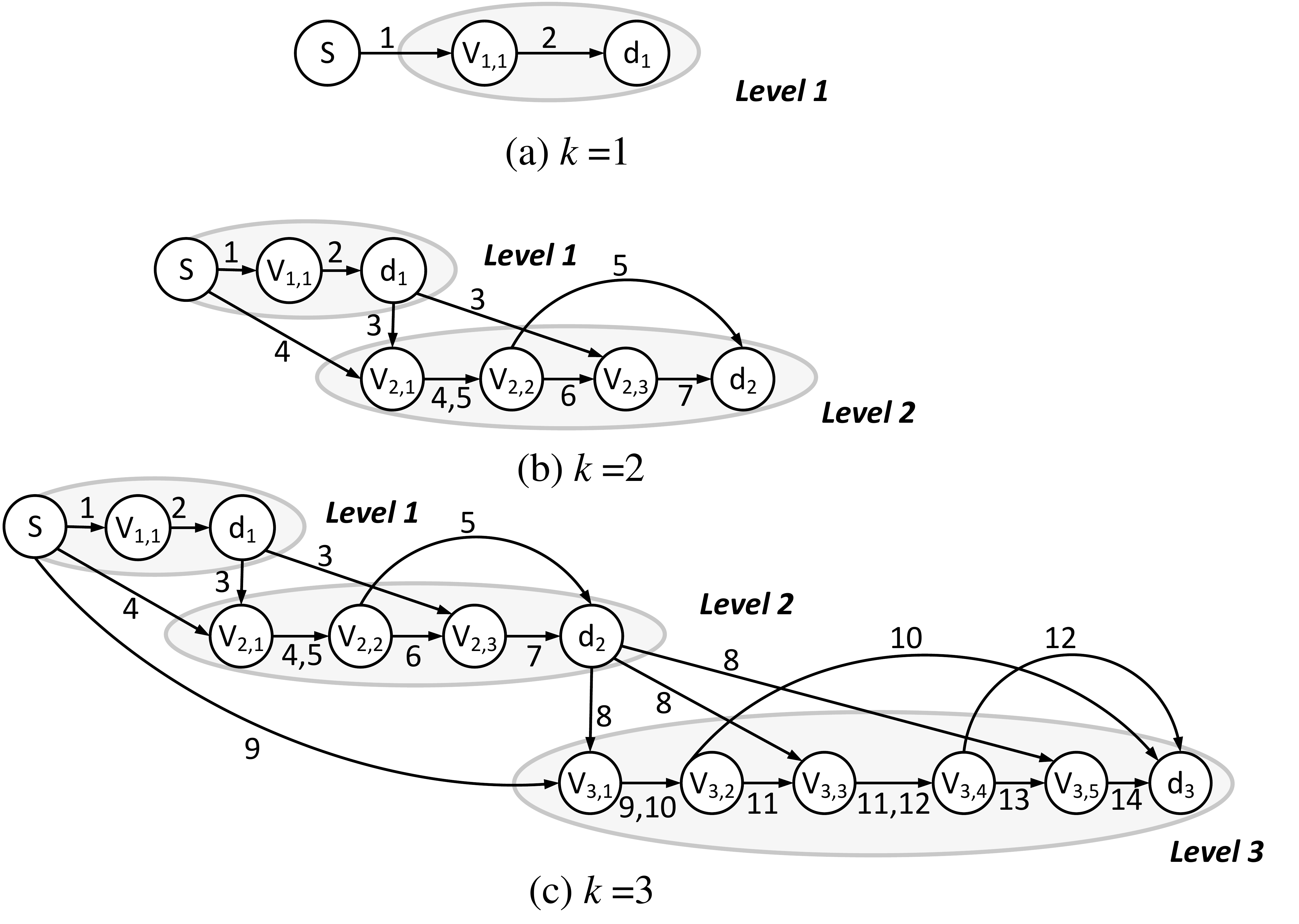}
\caption{Examples used in the proof to Property (II) in Theorem \ref{delta}. The source-destination pair is $(s,d_k)$ in graph $\mathcal{G}_k$ ($k=1,2,3$). Edge traversal delay is one slot.}
\label{bd}\vspace{-5mm}
\end{center}
\end{figure}
\subsection{Proof to Theorem \ref{max_hardness}} \label{proof_max_hardness}
The proof is based on a reduction from the Bounded-Length Edge-Disjoint Paths (BLEDP) problem which is NP-hard \cite{hard}.

\begin{itemize}

\item \textsc{Problem:} BLEDP.

\item \textsc{Instance:}
\begin{itemize}
\item [--] {A weighted digraph $G'=(V',E')$, where the weight on edge $e$ indicates its length (denoted by $l_e$). The length of each edge is a positive integer.}
\item [--] {The source-destination pair $(s,d)$.}
\item [--] {A bounded integer $L>0$ indicating the length bound.}
\end{itemize}

\item \textsc{Question:} Find the maximum number of edge-disjoint paths from $s$ to $d$ in $G'$ such that the length of each of these paths is upper-bounded by $L$.
\end{itemize}

Here we make an additional assumption that there exists no edge with its length greater than $L$ in $G'$. We also assume that there are no isolated nodes in $G'$. These assumptions do not change the complexity of BLEDP because we can simply remove these isolated nodes or long edges from $G'$ without any influence on the optimal solution.

The high-level idea of the reduction is to transform the ``spatial length bound" into a ``temporal length bound". Note that in our model, a natural temporal bound $T$ exists so we set $T=L$. In addition, we also need to make sure that whenever  edge $e$ is crossed, a ``temporal distance" of $l_e$ slots is traversed. Since it is assumed that edge-traversal delay is one time slot, we can expand each edge in series such that extra delay is incurred. To be more specific, if the length of edge $e$ is $l_e$, we replace this single edge by $l_e$ edges that are catenated in series; each of the catenated edges has one-slot traversal delay and is active in the entire time span. An example is illustrated in Figure \ref{series}.  It is trivial to check that BLEDP is equivalent to solving $\delta$-MAXFLOW in the constructed time-varying graph for $\delta=T$. Hence, $\delta$-MAXFLOW is NP-hard.

It remains to investigate the hardness of approximation for $\delta$-MAXFLOW.  Guruswami \emph{et al.} \cite{hard} proved that it is NP-hard to achieve $O(\sqrt{|E'|})$-approximation for BLEDP. In the constructed time-varying graph, we have $|E|=\sum_{e\in E'}l_e\le L|E'|$. Since $L$ is a bounded integer,  it follows that $|E'|=\Omega(|E|)$. Therefore, it is also NP-hard to achieve $O(\sqrt{|E|})$-approximation for $\delta$-MAXFLOW.

\begin{figure}[tbp]
\begin{center}
\includegraphics[width=2.15in]{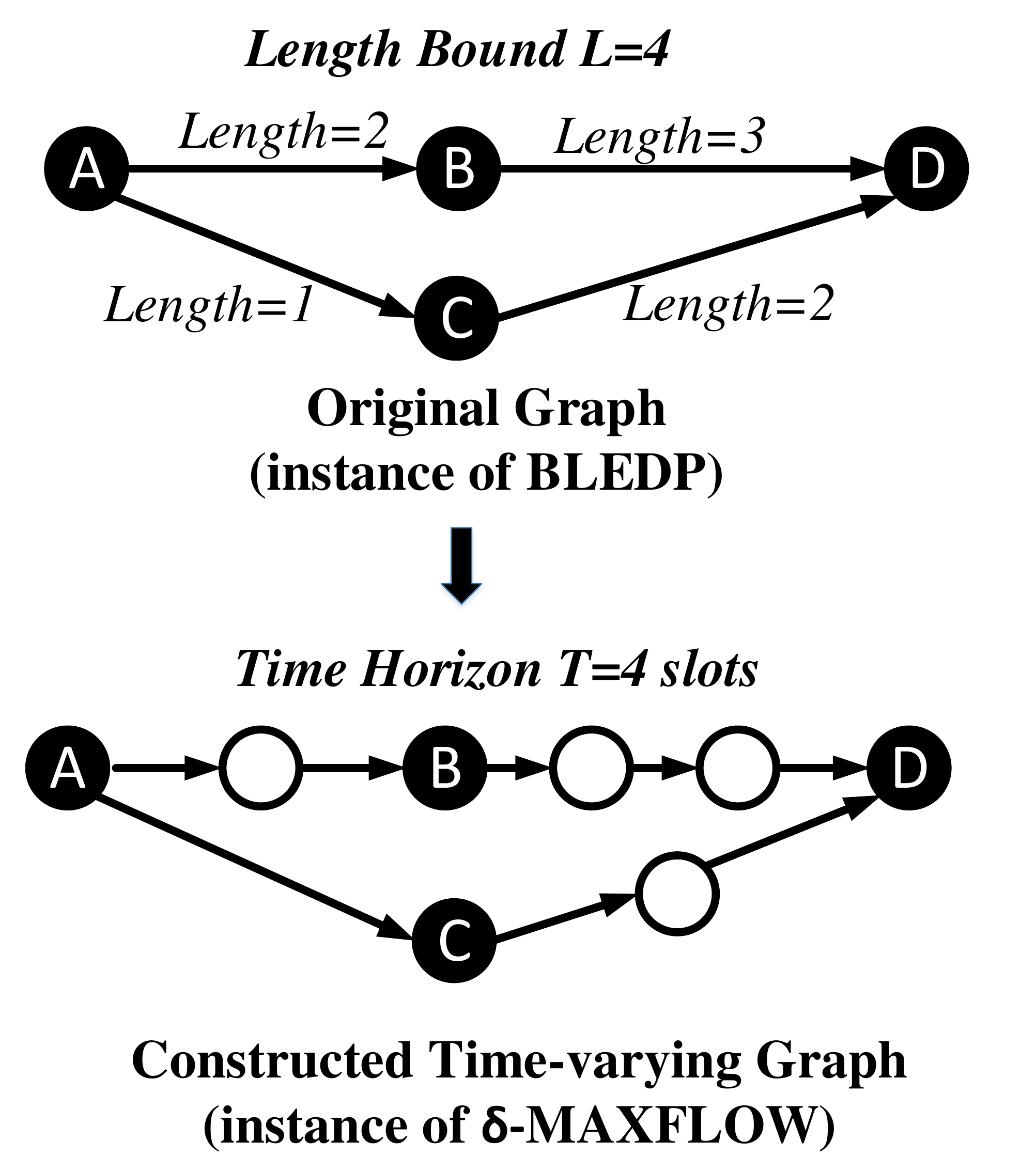}
\caption{Illustration of the reduction from BLEDP to $\delta$-MAXFLOW. Note that in the constructed time-varying graph, edge traversal delay is one time slot and each edge is active in the entire time span $\{1,2,3,4\}$}
\label{series}\vspace{-5mm}
\end{center}
\end{figure}
\subsection{Proof to Theorem \ref{ratio_maxflow}}\label{proof_ratio_maxflow}
If the the destination is unreachable from the source, both the optimal solution and the greedy algorithm will yield a result of zero, where no approximation gaps exist. Hence, it is enough to consider the scenario where the destination is reachable from the source.

Before the detailed proof, it is essential to define the notions of \emph{short paths} and \emph{long paths} in the Line Graph. Let $k$ be an \emph{arbitrary} positive integer. A short path consists of at most $k$ nodes while a long path is made up of more than $k$ nodes. Their corresponding journeys are called the \emph{short journey} (traversing at most $k$ edges) and the \emph{long journey} (traversing more than $k$ edges), respectively. Denote $\mathcal{J}^*=\{J^*_1,\cdots\}$ the optimal solution and $\mathcal{J}=\{J_1,\cdots\}$ the solution obtained by the greedy algorithm.

We first prove that the number of long journeys in $\mathcal{J}^*$ is at most $\frac{|E|(\frac{T}{\delta}+1)}{k}$.  Indeed, since journeys in $\mathcal{J}^*$ are $\delta$-disjoint, each edge can be traversed by at most $\lceil\frac{T}{\delta}\rceil$ journeys in $\mathcal{J}^*$. At the same time, each of the long journeys in $\mathcal{J}^*$ traverses more than $k$ edges so the total number of long journeys in $\mathcal{J}^*$ can be at most $\lfloor \frac{\lceil\frac{T}{\delta}\rceil|E|}{k}\rfloor\le \frac{|E|(\frac{T}{\delta}+1)}{k}$.
%Remark: a better bound is $\min\{\frac{|M|}{k},\frac{|E|T}{\delta k}\}$

Then we prove that the number of short journey in $\mathcal{J}^*$ is at most $2k\times |\mathcal{J}|$. To show this point, we first prove that each short journey (say $J_j^*$) in $\mathcal{J}^*$ is interfered by some short journey (say $J_i$) in $\mathcal{J}$ (i.e., $J^*_j$ and $J_i$ use the same edge within $\delta$ slots). Note that each short journey in $\mathcal{J}^*$ must be interfered by at least one journey in $\mathcal{J}$ otherwise the greedy algorithm is not finished. Let $J_i\in \mathcal{J}$ be the journey that interferes with some journey $J^*_j\in\mathcal{J}^*$ \emph{for the first time}, i.e., journeys constructed in the greedy algorithm before $J_i$ do not interfere with $J^*_j$. In other words, when the greedy algorithm is constructing journey $J_i$, journey $J^*_j$ is also a candidate journey. Since $J_i$ is selected rather than $J^*_j$, it implies that the number of edges traversed by $J_i$ is less or equal to that of $J^*_j$. Due to the fact that $J^*_j$ is a short journey, we can conclude that $J_i$ is also a short journey.

Meanwhile, each short journey in $\mathcal{J}$ can interfere with at most $2k$ $\delta$-disjoint journeys because any short journey in $\mathcal{J}$ contains at most $k$ contacts and each of these contacts can interferes with at most 2 $\delta$-disjoint journeys. Hence, the total number of $\delta$-disjoint journeys that can be interfered by the short journeys in $\mathcal{J}$ is at most $2k\times |\mathcal{J}|$. Since we have shown that each short journey in $\mathcal{J}^*$ is interfered by at least one short journey in $\mathcal{J}$, it is safe to conclude that the number of short journeys in $\mathcal{J}^*$ is upper-bounded by $2k\times |\mathcal{J}|$, which means that
\begin{equation}\label{ap}
|\mathcal{J}^*|=|\mathcal{J}^*_{long}|+|\mathcal{J}^*_{short}|\le \frac{|E|(\frac{T}{\delta}+1)}{k}+2k\times |\mathcal{J}|
\end{equation}
Now we set $k$ to be the integer such that $\sqrt{|E|(\frac{T}{\delta}+1)}\le k< \sqrt{|E|(\frac{T}{\delta}+1)}+1$. Then it follows that
\[
\begin{split}
|\mathcal{J}^*|&< \sqrt{|E|(\frac{T}{\delta}+1)}+2\Big(\sqrt{|E|(\frac{T}{\delta}+1)}+1\Big)|\mathcal{J}|\\
&\le \sqrt{|E|(\frac{T}{\delta}+1)}|\mathcal{J}|+2\Big(\sqrt{|E|(\frac{T}{\delta}+1)}+1\Big)|\mathcal{J}|\\
& = \Big(3\sqrt{|E|(\frac{T}{\delta}+1)}+2\Big)|\mathcal{J}|
\end{split}
\]
where the first inequality follows from the setting of $k$ and the second inequality holds because of our premise that $|\mathcal{J}|\ge 1$ (i.e., the destination is reachable from the source). Since $T$ is a bounded integer and $\delta\le T$, we can finally conclude that Algorithm \ref{alg_maxflow1} achieves $O(\sqrt{|E|})$-approximation.
\subsection{Proof to Theorem \ref{heu_bound}}\label{proof_heu_bound}
We  make two simple observations regarding the weights. The first is that $\omega_{e,t}\ge \frac{1}{\delta}$ since $K_{e,t}\le \delta$. The second is that the sum of weights that can be removed by one $\delta$-removal is less or equals to 1. Indeed, consider a certain $\delta$-removal that deletes contacts $(e,t_1),(e,t_2),\cdots,(e,t_n)$. It should be obvious that $K_{e,t_i}\ge n$ for any $1\le i\le n$, which means $\sum_{i=1}^n\omega_{e,t_i}=\sum_{i=1}^n \frac{1}{K_{e,t_i}}\le \sum_{i=1}^n \frac{1}{n}=1$. Then we introduce the following lemma.
\begin{lemma}\label{heu_lemma}
 Let $C$ be an arbitrary set of contacts whose removals disconnect the source-destination pair. The following result holds
\[
\sum_{(e,t)\in C}\omega_{e,t}\le |\mathsf{Cover_{\delta}}(C)|\le \delta\sum_{(e,t)\in C}\omega_{e,t}.
\]
\end{lemma}
\begin{proof}
The lower bound directly follows from the second observation mentioned above. Then we get down to proving the upper bound.

Denote $E_c$ the set of underlying edges in $C$. For each edge $e\in E_c$, suppose we need $n_e$ $\delta$-removals to completely delete $e$ from $C$, and the corresponding removal heads are $(e,t_1),(e,t_2),\cdots,(e,t_{n_e})$, where we assume $1\le t_1<t_2<\cdots<t_{n_e}\le T$. Denote $C_{e,i}$ the set of contacts deleted by the $\delta$-removal with head $(e,t_i)$ and define
 %Note that $t_{i+1}-t_i\ge \delta$ for any $1\le i\le n_e-1$, which implies that any two sets in the collection $\{C_{e,i}|e\in E_c, 1\le i\le n_e\}$ do not intersect.
\begin{equation}\label{pl}
W_{e,i}=\sum_{(e,t)\in C_{e,i}}\omega_{e,t},~\forall e\in E_c~\text{and}~1\le i\le n_e.
\end{equation}
Then we have
\begin{equation}\label{pq}
|\mathsf{Cover_{\delta}}(C)|=\sum_{e\in E_c}n_e=\sum_{e\in E_c}\sum_{i=1}^{n_e}\frac{\sum_{(e,t)\in C_{e,i}}\omega_{e,t}}{W_{e,i}},
\end{equation}
where the last equality is due to equation \eqref{pl}. We also notice that for any $e\in E_c$ and $1\le i\le n_e$
\[
\sum_{(e,t)\in C_{e,i}}\omega_{e,t}\ge \omega_{e,t_i},
\]
because contact $(e,t_i)$ is included in $\mathcal{C}_{e,i}$. By simple transformations, we obtain
\[
\begin{split}
\frac{\sum_{(e,t)\in C_{e,i}}\omega_{e,t}}{\omega_{e,t_i}}\ge 1=\frac{\sum_{ (e,t)\in C_{e,i}}\omega_{e,t}}{W_{e,i}}.
\end{split}
\]
Since $\omega_{e,t_i}\ge \frac{1}{\delta}$, we have
\[
\delta \sum_{(e,t)\in C_{e,i}}\omega_{e,t}\ge \frac{\sum_{(e,t)\in \mathcal{C}_{e,i}}\omega_{e,t}}{\omega_{e,t_i}} \ge \frac{\sum_{(e,t)\in C_{e,i}}\omega_{e,t}}{W_{e,i}}.
\]
Taking the above inequality into \eqref{pq}, we obtain
\[
|\mathsf{Cover_{\delta}}(C)|\le \delta \sum_{e\in E_c}\sum_{i=1}^{n_e}\sum_{(e,t)\in C_{e,i}}\omega_{e,t}=\delta\sum_{(e,t)\in C}\omega_{e,t},
\]
where the last equality holds because $C=\bigcup_{e\in E_c}\bigcup_{i=1}^{n_e}C_{e,i}$ and any two sets in the collection $\{C_{e,i}|e\in E_c, 1\le i\le n_e\}$ do not intersect.
\end{proof}

With the above lemma, we are ready to prove the approximation ratio for the min-weight algorithm. Suppose $C_{ALG}$ is the set of contacts disabled by the solution of the min-weight algorithm and $C^*$ is the set of contacts disabled by the optimal solution to $\delta$-MINCUT. Then according to Lemma \ref{heu_lemma}, we have
\[
\begin{split}
|\mathsf{Cover_{\delta}}(C_{ALG})|&\le \delta\sum_{(e,t)\in {C_{ALG}}}\omega_{e,t}.
\end{split}
\]
Since the min-weight algorithm first finds the minimum number of 1-removals that can disconnect the source-destination pair in the weighted time-varying graph, we have
\[
\sum_{(e,t)\in C_{ALG}}\omega_{e,t}\le  \sum_{(e,t)\in C^*}\omega_{e,t}.
\]
This implies that
\[
|\mathsf{Cover_{\delta}}(C_{ALG})| \le \delta \sum_{(e,t)\in C^*}\omega_{e,t}\le \delta|\mathsf{Cover_{\delta}}(C^*)|,
\]
where the last inequality is due to the lower bound in Lemma \ref{heu_lemma}. Therefore, $\delta$-approximation is achieved by the min-weight algorithm.

\begin{thebibliography}{1}

\bibitem{vehicular} X. Zhang, J. Kurose, B. N. Levine, D. Towsley, and H. Zhang, ``Study of a bus-based disruption-tolerant network: Mobility
modeling and impact on routing," in \emph{Proc. ACM MobiCom}, 2007.

\bibitem{trace}  A. Balasubramanianm, B. N. Levine, and A. Venkataramani, ``Enabling Interactive Applications for Hybrid Networks,"  \emph{ACM Mobicom}, 2008.

\bibitem{space1} J. Mukherjee and B. Ramamurthy, ``Communication technologies and architectures for space network and interplanetary Internet," in \emph{IEEE Commun. Surv. Tutorials}, vol. 15, no. 2, pp. 881–897, 2013.

\bibitem{space2} S. Burleigh and A. Hooke, ``Delay-tolerant networking: An approach to interplanetary internet," in \emph{IEEE Commun. Mag.}, vol. 41, no. 6, pp. 128–136, pp. 128–136, 2003.

\bibitem{mobile-sensor1} Y. Wang, H. Dang, and H. Wu, ``A survey on analytic studies of delay-tolerant mobile sensor networks: Research articles," in \emph{Wirel. Commun. Mob. Comput.}, vol. 7, no. 10, pp. 1197–1208, 2007.

\bibitem{mobile-sensor2} P. Juang, H. Oki, Y. Wang, M. Martonosi, L.S. Peh, and D. Rubenstein, ``Energy-efficient computing for wildlife tracking: Design
tradeoffs and early experiences with Zebranet," in \emph{ACM SIGOPS Oper. Syst. Rev.}, vol. 36, no. 5, pp. 96–107, 2002.

\bibitem{whitespace1} P. Bahl, R. Chandra, T. Moscibroda, R. Murty, and M. Welsh, ``White Space Networking with Wi-Fi like Connectivity," in \emph{ACM SIGCOMM}, 2009.

\bibitem{whitespace2} Bozidar Radunovic, Ranveer Chandra, and Dinan Gunawardena, ``Weeble: Enabling Low-Power Nodes to Coexist with
High-Power Nodes in White Space Networks," in \emph{ACM CoNext}, 2012.

\bibitem{whitespace3} Q. Liang, X. Wang, X. Tian, F. Wu, and Q. Zhang, ``Two-Dimensional Route Switching in Cognitive Radio Networks: A Game-Theoretical Framework," in \emph{IEEE/ACM Transactions on Networking}, vol. 23, no. 4, pp. 1053-1066, 2015.

\bibitem{mmWave} S. Rangan, T. Rappaport, and E. Erkip, ``Millimeter Wave Cellular Wireless Networks: Potentials and Challenges," in \emph{Proceedings of the IEEE}, vol. 102, no. 3, pp.366-385, 2014.
    
\bibitem{duty1} Y. Gu and T. He, ``Dynamic switching-based data forwarding for low-duty-cycle wireless sensor networks," in \emph{IEEE Trans. Mobile Comput.}, vol. 10, no. 12, pp. 1741–1754, Oct. 2011.

\bibitem{duty2} L. Chen, Y. Gu, S. Guo, T. He, Y. Shu, F. Zhang, and J. Chen, ``Group-based discovery in low-duty-cycle mobile sensor
networks," in \emph{Proc. IEEE Commun. Soc. Conf. on Sensor Mesh Ad Hoc Commun. Netw.}, pp. 542–550, 2012.

\bibitem{databse} FCC. Order, FCC 11-131, 2011.

\bibitem{human-mobility} C. Song, Z. Qu, N. Blumm, and A. Barabasi, ``Limits of predictability in human mobility," in \emph{Science}, vol. 327, pp. 1018–1021, 2010.

\bibitem{mobile-social} V. Srinivasan, M. Motani, and W.T. Ooi, ``Analysis and implications of student contact patterns derived from campus schedules," in \emph{Proc. ACM Mobicom}, 2006.

\bibitem{sensor1} Y. Gu and T. He, ``Dynamic switching-based data forwarding for low-duty-cycle wireless sensor networks," in \emph{IEEE Trans. Mobile Comput.}, vol. 10, no. 12, pp. 1741–1754, Oct. 2011.

\bibitem{sensor2} F. Li, S. Chen, S. Tang, X. He, and Y. Wang, ``Efficient topology design in time-evolving and energy-harvesting wireless sensor networks," in \emph{Proc. IEEE Mob. Ad Hoc Sensor Syst.}, 2013, pp. 1–9.


\bibitem{TVG1} J. Whitbeck, M. Amorim, V. Conan, and J. Guillaume, ``Temporal Reachability Graphs,"  \emph{ACM Mobicom}, 2012.
\bibitem{TVG2} A. Casteigts, P. Flocchini, W. Quattrociocchi, and N. Santoro, ``Time-varying graphs and dynamic networks,"  \emph{Ad-hoc, Mobile, and Wireless Networks}, vol. 6811, pp. 346-359, 2011.
\bibitem{TVG3} B. Xuan, A. Ferreira, and A. Jarry, ``Computing the shortest, fastest, and foremost journeys in dynamic networks,"  \emph{International Journal of Foundations of Computer Science}, vol. 14, pp. 267-285, 2003.
\bibitem{TVG4}  S. Scellato, I. Leontiadis, C. Mascolo, P. Basuy, and M. Zafer, ``Evaluating Temporal Robustness of Mobile Networks,"  \emph{IEEE Transactions on Mobile Computing}, vol. 12,  no. 1, pp. 105-117, 2013.


\bibitem{LP} Dimitris Bertsimas and John N. Tsitsiklis. \emph{Introduction to Linear Optimization}. Athena Scientific, 1997.

\bibitem{vul} K. A. Berman, ``Vulnerability of Scheduled Networks and a Generalization of Menger's Theorem,"  in \emph{Networks}, John Wiley \& Sons, vol. 28, pp. 125-134, 1996.
\bibitem{con} D. Kempe, J. Kleinberg, and A. Kumar, ``Connectivity and Inference Problems for Temporal Networks," \emph{ACM STOC}, 2000.
\bibitem{hard} V. Guruswami, S. Khanna, R. Rajaraman, B. Shepherd, and M. Yannakakis, ``Near-Optimal Hardness Results and Approximation Algorithms for Edge-Disjoint Paths and Related Problems," \emph{ACM STOC}, 1999.
%\bibitem{membership}    N. Robertson and P. D. Seymour, ``Graph minors. XIII. The disjoint paths problem,"  \emph{Journal of Combinatorial Theory, Series B}, vol. 63, pp. 65-110, 1995.
%\bibitem{trace2} Shared Spectrum Company. Survey of Radio Frequency Bands: 30 MHz to 3 GHz.  Spectrum Reports, 2010.
%\bibitem{julia}    M. Lubin and I. Dunning. Computing in Operations Research using Julia. arXiv:1312.1431, 2013.
%\bibitem{switching} Q. Liang, X. Wang, X. Tian, F. Wu, and Q. Zhang. Two-Dimensional Route Switching in Cognitive Radio Networks: A Game-Theoretical Framework.  \emph{IEEE/ACM Transactions on Networking (ToN)}, 2014.
%\bibitem{block}  Habiba, Y. Yu, T. Y. B. Wolf, and J. Saia. Finding Spread Blockers in Dynamic Networks.  \emph{Proceedings of ACM SNA-KDD}, 2008.
%\bibitem{inhibition}  C. A. Phillips. The network inhibition problem.  \emph{Proceedings of ACM STOC}, 1993.
%\bibitem{k_short} D. A. Dunn, W. D. Grover, and M. H. MacGregor. Comparison of k-shortest paths and maximum flow routing for network facility restoration.  \emph{IEEE Journal on Selected Areas in Communications}, vol. 12, no. 1, pp. 88-99, 1994.
%\bibitem{pebbling_game_ref}    S. Fortune, J. Hopcroft, and J. Wyllie. The directed subgraph homeomorphism problem.  \emph{Theoretical Computer Science} vol. 10, pp. 111-121, 1980.
\bibitem{dis_protect1} A. Srinivas and E. Modiano, ``Minimum Energy Disjoint Path Routing in Wireless Ad Hoc Networks,"  \emph{ACM Mobicom}, 2003.
\bibitem{dis_protect2}  G. Kuperman and E. Modiano, ``Disjoint Path Protection in Multi-Hop Wireless Networks with Interference Constraints,"  \emph{IEEE INFOCOM}, 2013.
\bibitem{redundancy} S. Jain, M. Demmer, R. Patra, and K. Fall, ``Using Redundancy to Cope with Failures in a Delay Tolerant Network,"  \emph{ACM SIGCOMM}, 2005.
\bibitem{DTN_routing1} S. Jain, K. Fall, and R. Patra, ``Routing in a delay-tolerant network,"  \emph{ACM SIGCOMM}, 2004.
\bibitem{DTN_routing2} A. Vahdat and D. Becker, ``Epidemic routing for partially connected ad hoc networks,"  Technical Report, Department of Computer Science, Duke University, 2000.
\bibitem{DTN_routing3} A. Oria, and O. Scheln, ``Probabilistic routing in intermittently connected networks,"  \emph{ACM MobiHoc}, 2003.
%\bibitem{DTN_routing4} J. Burgess, B. Gallagher, D. Jensen, and B. N. Levine. MaxProp: Routing for vehicle-based disruption-tolerant networks.  \emph{Proc. IEEE INFOCOM}, 2006.
\bibitem{TVG-unreliable} F. Li, S. Chen, M. Huang, Z. Yin, C. Zhang, and Y. Wang, ``Reliable Topology Design in Time-Evolving Delay-Tolerant Networks with Unreliable Links," \emph{IEEE Transactions on Mobile Computing,} vol. 14, no. 6, pp. 1301-1314, 2015.
%\bibitem{perfect_graph} Chudnovsky, Maria, Neil Robertson, Paul Seymour, and Robin Thomas, ``The strong perfect graph theorem," \emph{Annals of Mathematics}, pp. 51-229, 2006.
%\bibitem{TVG5} A. Ferreira. Building a reference combinatorial model for MANETs.  \emph{IEEE Networks}, vol. 18, no. 5, pp. 24-29, 2004.
\bibitem{TVG6} H. Kim and R. Anderson, ``Temporal node centrality in complex networks,"  \emph{Physical Review E}, vol. 85, no. 2, id. 026107, 2012.
\bibitem{TVG7} R. Pan and J. Saramaki, ``Path lengths, correlations, and centrality in temporal networks,"  \emph{Physical Review E}, vol. 84, no. 1, id. 016105, 2011.
\bibitem{TVG8} A. Chaintreau, A. Mtibaa, L. Massoulie, and C. Diot, ``The diameter of opportunistic mobile networks,"  \emph{ACM CoNEXT}, 2007.
\bibitem{TVG9} J. Leskovec, J. Kleinberg, and C. Faloutsos, ``Graphs over time: densification laws, shrinking diameters and possible explanations,"  \emph{ACM SIGKDD}, 2005.
\bibitem{TVG10} S. Pierre, M. Barbeau, and E. Kranakis, ``Complexity of Connected Components in Evolving Graphs and the Computation of Multicast Trees in Dynamic Networks,"  \emph{Ad-Hoc, Mobile, and Wireless Networks}, vol. 2865, pp. 259-270, 2003.

\bibitem{TVG11} A. Ferreira1 and A. Jarry, ``Minimum-Energy Broadcast Routing in Dynamic Wireless Networks,"  \emph{Journal of Green Engineering}, vol. 2, no. 2, pp. 115-123, 2012.

\bibitem{line_ref} L. W. Beineke, ``Characterizations of derived graphs,"  \emph{Journal of Combinatorial Theory}, vol. 9, no. 2, pp. 129-135, 1970.

\end{thebibliography}
\end{document}